\documentclass[graybox,envcountsect,envcountsame]{svmult}
\usepackage[
    backend=biber,
    style=alphabetic,
    url=false,
    doi=false,
    isbn=false,
    backref=true,
    maxbibnames=99,
    useprefix=false,
]{biblatex}

\newcommand{\doiorurl}{%
  \iffieldundef{doi}
    {\iffieldundef{url}
       {}
       {\strfield{url}}}
    {http://dx.doi.org/\strfield{doi}}%
}

\newcommand{\myhref}[1]{%
 \ifboolexpr{%
   test {\ifhyperref}
   and
   not test {\iftoggle{bbx:url}}
   and
   not test {\iftoggle{bbx:doi}}
  }
  {\href{\doiorurl}{#1}}
  {#1}%
}

\DeclareFieldFormat{title}{\myhref{\mkbibemph{#1}}}
\DeclareFieldFormat
  [article,inbook,incollection,inproceedings,patent,thesis,unpublished]
  {title}{\myhref{\mkbibquote{#1\isdot}}}

\usepackage{orcidlink}
\usepackage{type1cm}

\usepackage{makeidx}         
\usepackage{graphicx}        

\usepackage{multicol}        
\usepackage[bottom]{footmisc}

\usepackage{newtxtext}       %
\usepackage[varvw]{newtxmath}       

\usepackage{tikz}
\usepackage[x11names]{xcolor}
\usetikzlibrary{decorations.pathreplacing,calligraphy}
\usetikzlibrary{positioning}
\usetikzlibrary{decorations}
\usetikzlibrary{decorations.pathmorphing}
\usetikzlibrary{shapes, fit, arrows.meta}
\usetikzlibrary{calc}
\usepackage{pgfplots}
\usepgfplotslibrary{groupplots}
\pgfplotsset{compat=1.18}

\usepackage[activate={true,nocompatibility},final,tracking=true,kerning=true,spacing=true,factor=1100,stretch=10,shrink=10]{microtype}

\makeindex

\def\Caratheodory{Carath\'{e}odory}
\newcommand{\Conv}{\operatorname{Conv}}

\usepackage{mathtools, stmaryrd}
\usepackage{dirtytalk}

\usepackage{hyperref}
\usepackage{listings}
\usepackage{etoolbox}

\hypersetup{
    hidelinks,
    colorlinks,
    hypertexnames=false,
    linkcolor=[rgb]{0.3,0.3,0.6},
    citecolor=[rgb]{0.2, 0.6, 0.2},
    urlcolor=[rgb]{0.6, 0.2, 0.2}
}

\usepackage[nameinlink,capitalise]{cleveref}

\usepackage{macros}

\begin{document}

\title*{A primer on the closure of algebraic complexity classes under factoring}

\author{C.S. Bhargav\orcidlink{0000-0002-6920-4998} and Prateek Dwivedi\orcidlink{0000-0002-0572-3721} and Nitin Saxena\orcidlink{0000-0001-6931-898X}}
\institute{C.S. Bhargav \at University of Regensburg, Germany, and IIT Kanpur, India, \email{bhargav.cs@ur.de}
\and Prateek Dwivedi \at ITU Copenhagen, Denmark, \email{prdw@itu.dk}
\and Nitin Saxena \at IIT Kanpur, India, \email{nitin@cse.iitk.ac.in}}

{\def\addcontentsline#1#2#3{}\maketitle}
\renewcommand{\leftmark}{Closure of algebraic complexity classes}
\renewcommand{\rightmark}{Bhargav, Dwivedi \& Saxena}

\abstract{
    Polynomial factorisation is a fundamental problem in computational algebra. Over the past half century, a variety of algorithmic techniques have been developed to tackle different variants of this problem. In parallel, algebraic complexity theory classifies polynomials into complexity classes based on their computational hardness. This raises a natural question: Are these complexity classes closed under factorisation?\\
    \indent
    In this survey, we revisit pivotal techniques in polynomial factorisation: Hensel lifting, Newton iteration, and Lagrange inversion. These techniques have played an essential role in resolving key factoring questions in algebraic complexity for more than half a century. We examine and organise the known results through the lens of these techniques, discussing their underlying mathematical equivalence while reflecting on how their applications vary depending on the problem context.\\
    \indent
    We focus on prominent algebraic complexity classes, including $\mathsf{VP}$ (circuits of polynomial size and degree), its closure $\overline{\mathsf{VP}}$, the class $\mathsf{VNP}$ (verifier circuits of polynomial size and degree), $\mathsf{VBP}$ (polynomial-size branching programs), $\mathsf{VF}$ (polynomial-size formulas), and $\mathsf{VP}_{\text{nb}}$ (circuits of polynomial size and exponential degree). We also discuss bounded-depth circuits and sparse polynomials. Along the way, we highlight several unresolved open problems.  
}

\vspace{5cm}
\begingroup
\let\clearpage\relax 
\setcounter{tocdepth}{2}
\tableofcontents
\endgroup
\clearpage
\section{Introduction}
\label{sec:intro}

The problem of finding a nontrivial factor of a polynomial is a classical and fundamental problem, with a rich history spanning centuries~\cite{von2006}. Kaltofen \cite{Kal1982,kaltofen1990,Kal1992} gave a thorough treatment of several foundational results, while von zur Gathen and Panario \cite{vP2001} focused on factorisation over finite fields. 
These developments are also covered in greater depth in standard textbooks such as~\cite{vG2013} and~\cite{Sho2009}.

In this survey, we focus on the core techniques and ideas that drive factorisation results within algebraic circuit complexity. Our aim is to illustrate how these tools have contributed to resolving long-standing questions in the field. Notably, Forbes and Shpilka~\cite{FS2015} have provided an accessible exposition of the high-level ideas behind some factoring algorithms, emphasising their relevance to other problems in algebraic complexity. Building on this, our survey aims to delve into the technical intricacies of these results by organising them according to the underlying methods. More than algorithms, we will be concerned with whether the factors of a `structured' polynomial are also structured. We begin with a gentle introduction to each technique, followed by a discussion of the key results that it enables. The simplest case, perhaps, is when the polynomial is of a single variable.

\subsection{Univariate factoring} \label{sec:univariate-factor} 

The polynomial ring $\F[x]$ over a field $\F$ is a well-known unique factorisation domain. This ensures that any non-constant polynomial $f \in \F[x]$ can be uniquely decomposed into the product of a constant and a collection of irreducible polynomials. Consequently, computing this decomposition presents a very natural computational challenge. The quest to factor polynomials traces its roots back to ancient Babylonian times, and it has consistently been considered a problem of profound importance due to its extensive theoretical and practical applications.

\begin{questype}{Problem} 
\label{prob:uni-factor}
    \noindent Given a univariate polynomial $f(x)$ over a field $\F$, compute pairwise distinct irreducible polynomials $f_1, \ldots, f_r \in \F[x]$ and a tuple $(e_1, \ldots, e_r) \in \N^r$ such that
    $f \;=\; f_1^{e_1} \cdots f_r^{e_r}$.
\end{questype}

Considered over the rational numbers $\Q$, the polynomial $x^2 - 2$ is irreducible, whereas it factors as $(x-3)(x-4) \bmod 7$. Evidently, the problem critically depends on the field $\F$. We begin by considering factorisation over a \emph{finite field} $\F_q$ of order $q = p^a$, for some prime $p$. Despite its simplicity, this case involves several non-trivial ideas that serve as the foundation for more general algorithms, including those over the rationals and algebraic number fields.

The input polynomial $f(x) = \sum_{i=0}^{d} c_i x^i \in \F_q[x]$ is given in the \emph{dense representation} as a list of (coefficient, exponent) pairs $(c_i, i)$ for all $0 \leq i \leq d$. Let $t \in \F_p[y]$ be some irreducible polynomial of degree $a$, such that $\F_q \cong \F_p[y]/\langle t(y) \rangle$. Given such a $t$, arithmetic operations in $\F_q$ can be performed in time $\poly(\log q)$. Hence, we assume that the irreducible polynomial $t$ is part of the input to the algorithm. A thorough exposition of computational issues in finite fields is available in~\cite{GS2013} and \cite[Section A.4]{AB2009}. The goal is then to solve univariate factoring over $\F_q$ in time polynomial in $d$ and $\log q$. We will be describing two classical algorithms, each offering a different perspective on the problem.

An important subroutine that underlies most of the factoring algorithms is computing the greatest common divisor (GCD) of two polynomials.
The classical Euclidean algorithm can be used to efficiently compute the $\gcd$~\cite[Section 3]{vG2013}. 
Throughout this text, we use the notation $\poly(n)$ to denote any function of the form $O(n^c)$ for some constant $c$. In particular, when we say that an algorithm takes $\poly(n)$ operations over a field $\F$, we refer to a count of field-level arithmetic operations performed within $\F$. This complexity measure abstracts away the bit-level details of the field operations, which depend on the representation of $\F$.

\begin{lemma}[GCD] \label{lem:gcd}
    Let $\F$ be a field and $f, g \in \F[x]$ be polynomials of degree at most $d$. Then, $\gcd(f,g)$ can be computed in $\poly(d)$ operations over $\F$. Moreover, the algorithm returns $a, b \in \F[x]$ such that $a \cdot f + b \cdot g = \gcd(f,g)$. 
\end{lemma}

Most factoring algorithms over finite fields, and beyond, follow a standard three-step structure (see \Cref{fig:fact-steps}):

\begin{enumerate}
    \item \textbf{Square-free Factorisation:}  
    Note that each $e_i \geq 1$ in \Cref{prob:uni-factor}. The \emph{square-free} part of $f$, denoted by $\operatorname{rad}(f) = f_1 \cdots f_r$, is called the \emph{radical} of $f$. 
    This can be efficiently obtained by dividing $f$ by $\gcd\left(f, \partial_x f\right)$. If $\partial_x f=0$ then $f$ is a perfect $p$-th power. We can apply the Frobenius automorphism to extract its $p$-th root, effectively computing $f^{1/p}$.

    \item \textbf{Distinct Degree Factorisation:}
    The goal is to decompose $\operatorname{rad}(f)$ into a product $\tilde{f}_1 \cdots \tilde{f}_d$, where each $\tilde{f}_i$ is the product of all irreducible degree-$i$ factors of $f$. A classical result from field theory shows that $$\tilde{f}_i = \gcd(x^{q^i} - x, f),$$ which can be computed efficiently using repeated squaring and \Cref{lem:gcd}.
    The efficiency arises from two observations: (1) computing $\gcd(x^{q^i} - x, f)$ is equivalent to $\gcd(f, (x^{q^i} - x) \bmod f)$, and (2) the powers $x^{q^i} \bmod f$ can be computed iteratively via $$x^{q^i} \;\equiv\; (x^{q^{i-1}})^q \pmod f.$$
    In practice, for uniformly random polynomials of high degree, this step often dominates the running time of the overall factoring process.

    \item \textbf{Equal Degree Factorisation:}  
    After the above preprocessing, we are left with the task of factoring a polynomial into irreducibles of the same degree. Two classical algorithms for this purpose are Cantor--Zassenhaus and Berlekamp’s algorithm. While Berlekamp's method is deterministic, Cantor--Zassenhaus is randomised and tends to scale better for large finite fields. Both approaches fundamentally exploit the structure given by the Chinese Remainder Theorem.
\end{enumerate}

\begin{figure}
\centering
\begin{tikzpicture}[
    every node/.style={font=\tiny},
    box/.style={rectangle, draw, minimum height=0.6cm, minimum width=0.5cm, align=center},
    >=Stealth
]
\begin{scope}[xshift=1.7cm]  

    \definecolor{c1}{rgb}{0.85,0.7,0.85}
    \definecolor{c2}{rgb}{0.7,0.9,1.0}
    \definecolor{c3}{rgb}{0.7,1.0,0.7}
    \definecolor{c4}{rgb}{0.75,0.75,0.95}
    \definecolor{c5}{rgb}{1.0,0.8,0.75}

    \node[box, fill=c1, minimum width=0.9cm] (s1top) at (0,0.5) {$\left(x+1\right)$};
    \node[box, fill=c1, minimum width=0.9cm] (s1bot) at (0,0) {$\left(x+1\right)$};
    \node[box, fill=c2, minimum width=1.4cm] (b2) at ($(s1bot.east)+(0.7,0)$) {$\left(x^2+1\right)$};
    \node[box, fill=c3, minimum width=1.0cm] (b3) at ($(b2.east)+(0.5,0)$) {$\left(x^2+3\right)$};
    \node[rectangle, draw, fill=c4, minimum height=0.6cm, minimum width=1.6cm] (b4) at ($(b3.east)+(0.8,0)$) {$\left(x^3+x+1\right)$};
    \node[rectangle, draw, fill=c5, minimum height=0.6cm, minimum width=1.8cm] (b5) at ($(b4.east)+(0.9,0)$) {$\left(x^4+1\right)$};

    \node[box, fill=c1, minimum width=0.9cm] (s1sep2) at (0,-1.5) {$\left(x+1\right)$};
    \node[box, fill=c1, minimum width=0.9cm] (s1sep1) at (0,-2.2) {$\left(x+1\right)$};
    \node[box, fill=c2, minimum width=1.4cm] (b2d) at (2.0,-2.0) {$\left(x^2+1\right)$};
    \node[box, fill=c3, minimum width=1.0cm] (b3d) at ($(b2d.east)+(0.5,0)$) {$\left(x^2+3\right)$};
    \node[rectangle, draw, fill=c4, minimum height=0.6cm, minimum width=1.6cm] (b4d) at ($(b3d.east)+(0.8,0)$) {$\left(x^3+x+1\right)$};
    \node[rectangle, draw, fill=c5, minimum height=0.6cm, minimum width=1.8cm] (b5d) at ($(b4d.east)+(0.9,0)$) {$\left(x^4+1\right)$};

    \coordinate (xbase) at (-1.2,-4.0);
    
    \node[box, fill=c2, minimum width=1.4cm] (b2dd) at ($(xbase)+(1.0,0)$) {$\left(x^2+1\right)$};
    \node[box, fill=c3, minimum width=1.0cm] (b3dd) at ($(b2dd.east)+(0.5,0)$) {$\left(x^2+3\right)$};
    
    \node[rectangle, draw, fill=c4, minimum height=0.6cm, minimum width=1.6cm] (b4dd) at ($(xbase)+(5.0,0)$) {$\left(x^3+x+1\right)$};
    \node[rectangle, draw, fill=c5, minimum height=0.6cm, minimum width=1.8cm] (b5dd) at ($(xbase)+(7.3,0)$) {$\left(x^4+1\right)$};

    \node[box, fill=c2] (b2ddd) at ($(xbase)+(0.5,-2.0)$) {$\left(x^2+1\right)$};
    \node[box, fill=c3] (b3ddd) at ($(xbase)+(2.0,-2.0)$) {$\left(x^2+3\right)$};

    \draw[->] ([yshift=-0.1cm]s1bot.south) -- ([yshift=0.1cm]s1sep2.north);
    \draw[->] ([yshift=-0.1cm]b4.south) -- ([yshift=0.1cm]b4d.north);
    
    \draw[->] ([yshift=-0.1cm]b2d.south) -- ([yshift=0.1cm]b2dd.north);
    \draw[->] ([yshift=-0.1cm]b4d.south) -- ([yshift=0.1cm]b4dd.north);
    \draw[->] ([yshift=-0.1cm]b5d.south) -- ([yshift=0.1cm]b5dd.north);
    
    \draw[->] ([yshift=-0.1cm]b2dd.south) -- ([yshift=0.1cm]b2ddd.north);
    \draw[->] ([yshift=-0.1cm]b3dd.south) -- ([yshift=0.1cm]b3ddd.north);

    \node[right, align=left] at (7.2,-0.8) {Square-free\\Factorisation};
    \node[right, align=left] at (7.2,-3.2) {Distinct Degree\\Factorisation};
    \node[right, align=left] at (7.2,-5.2) {Equal Degree\\Factorisation};

\end{scope}
\end{tikzpicture}
\caption{Polynomial factorisation steps. The input polynomial has two degree-1 factors, two degree-2 factors, one degree-3 factor, and one degree-4 factor. Adapted from \cite{vS1992}.}
\label{fig:fact-steps}
\end{figure}

In the remainder of this section, we focus on algorithms for \emph{Equal Degree Factorisation}. After performing the two preprocessing steps, we are left with factoring a polynomial of the form  $f = f_1 \cdot f_2 \cdots f_r$,
where each $f_i$ is an irreducible polynomial of the same degree.

\subsubsection{Berlekamp's algorithm}

At the core of Berlekamp's algorithm lies a fundamental number-theoretic observation: suppose there exists a polynomial $h \in \F_q[x]$ such that
$$
h^p - h \equiv 0 \pmod{f},
$$
and $1 \leq \deg(h) < \deg(f)$. Then, by the identity $h^p - h = \prod_{\alpha \in \F_q} (h - \alpha)$, the polynomial $f$ can be factored by computing $\gcd(f, h - \alpha)$ for all $\alpha \in \F_q$.

To see that such a non-constant $h$ exists, consider the Chinese Remainder Theorem isomorphism $\F_q[x]/\langle f \rangle \cong \bigoplus_{i=1}^r \F_q[x]/\langle f_i \rangle$, where the $f_i$ are the distinct irreducible factors of the square-free $f$ ($r \geq 2$).
Any polynomial $h$ mapping to $(\alpha_1, \dots, \alpha_r)$ under this isomorphism, with each $\alpha_i \in \F_p$, satisfies $h^p \equiv h \pmod{f}$.
Such an $h$ is non-constant precisely when not all $\alpha_i$ are equal.

To find such an $h$, the algorithm considers the $\F_p$-vector space
\begin{equation}
    \label{eq:berlekamp_space}
    V \;\coloneqq\; \setdef{h(x)}{\deg(h) \leq d, h^p - h =0 \pmod f}.
\end{equation}  
Recall $q = p^a$. This is a subspace of $\F_q[x]/\langle f \rangle$ of dimension at most $\deg(f) \cdot a$, and a basis for $V$ can be efficiently computed using standard linear algebra over $\F_p$. Any non-trivial (i.e., non-constant) basis element $h \in V$ yields a non-trivial factor of $f$ via GCD computations.

Notably, $q$-many GCD computations would be inefficient. To avoid that, elements of $\F_q$ are represented as vectors over $\F_p$.
Notably, Berlekamp’s algorithm does not require the input polynomial to have irreducible factors of the same degree; thus, the distinct-degree factorisation step can be bypassed. With careful implementation, the total runtime is bounded by $\poly(p, d, \log q)$, making it a deterministic polynomial-time algorithm in small characteristic, e.g., when $p = (d a)^{O(1)}$. For an accessible exposition, see the lecture notes by Kopparty~\cite[Lecture~10]{Kop2014}.

\subsubsection{Cantor--Zassenhaus Algorithm}\label{sec-CZ}

As described earlier, Berlekamp's algorithm is an efficient deterministic algorithm for fields of small characteristic.
By contrast, the Cantor--Zassenhaus algorithm leverages randomness to offer a more efficient alternative for factoring over large-characteristic fields, such as when $p = (da)^{\omega(1)}$.
The key idea is to choose a polynomial $h \in \F_q[x]$ of degree at most $d-1$ uniformly at random, and check whether the following $\gcd$ yields a non-trivial factor:
$$
\gcd\left(f,\, h^{\frac{q^d - 1}{2}} - 1 \bmod f\right).
$$

To see why this works, recall that $f = f_1 \cdot f_2 \cdots f_r$, obtained after the preprocessing steps. Since each $f_i$ is irreducible of degree at most $d$, the Chinese Remainder Theorem yields an isomorphism:
$$
\frac{\F_q[x]}{\inangle{f}} \cong \frac{\F_q[x]}{\inangle{f_1}} \times \cdots \times \frac{\F_q[x]}{\inangle{f_r}}.
$$
Consider a random polynomial $h \in \F_q[x]/\inangle{f}$ that maps to $(h_1, \dots, h_r)$ under this isomorphism, with the condition that $\gcd(f, h) = 1$ (i.e., none of the $h_i$ are zero, which would already yield a factor of $f$).
Then, $h^{q^d} - h = 0$. Assuming $q$ is odd, it follows that
$$
h \cdot \left(h^{\frac{q^d - 1}{2}} - 1\right) \cdot \left(h^{\frac{q^d - 1}{2}} + 1\right) = 0.
$$

The identity $h^{q^d} - h = \prod_{\alpha \in \mathbb{F}_{q^d}} (h - \alpha)$ implies that, excluding $0$, exactly half the elements of $\mathbb{F}_{q^d}$ satisfy $h^{\frac{q^d - 1}{2}} = 1$ and the other half satisfy $h^{\frac{q^d - 1}{2}} = -1$. Under the CRT map, the components of $h' \coloneqq h^{\frac{q^d - 1}{2}} - 1 \pmod f$ are therefore distributed such that each $h'_i \in \{0, -2\}$. 

A non-trivial factor is obtained if $h'$ contains a mix of both $0$ and $-2$ entries. 
If all $h'_i = 0$, then $h' \equiv 0 \pmod f$ and $\gcd(f, h')$ yields the trivial factor $f$. Similarly, if all $h'_i = -2$, then $h'$ is coprime to $f$ and the $\gcd$ yields $1$.
Since the components are chosen independently and at random, the probability of obtaining such a mixed vector is $1 - (1/2)^{r-1} \ge 1/2$ (for $r \ge 2$). Thus, computing $\gcd(f, h^{\frac{q^d - 1}{2}} - 1)$ yields a non-trivial factor of $f$ with high probability.

When $q$ is even, we instead compute $h' = h + h^2 + h^{2^2} + \cdots + h^{2^{rd - 1}}$ and check whether $\gcd(f, h')$ is non-trivial. This approach works because one can show that $h'(h' + 1) = h^{2^{rd}} + h$, enabling a similar probabilistic argument. The algorithm relies primarily on computing $\gcd$s and repeated squaring. The probabilistic bounds imply that, in expectation, only a constant number of repetitions (at most two) are required. Thus, the overall complexity of the Cantor--Zassenhaus algorithm is $\poly(d, \log q)$.

    



\begin{questype}{Problem}
    \noindent Devise a deterministic algorithm that, given a univariate polynomial $f \in \mathbb{F}_q[x]$ of degree $d$, computes its complete factorisation in time polynomial in $d$ and $\log q$.
\end{questype}

\begin{remark}
    It is unknown whether computing square roots modulo a prime $p$ (i.e., determining $x$ such that $x^2 \equiv a \pmod p$) is achievable in deterministic time polynomial in $\log p$. Readers are referred to \cite{IKRS2012} for a detailed survey on this line of work.
\end{remark}


\subsection{Fine-grained developments}

The complexity of univariate factoring via the Cantor--Zassenhaus algorithm is primarily dominated by modular exponentiation, requiring $O(d^{2+ o(1)} \log q)$ operations in $\mathbb{F}_q$. By optimising the Equal Degree Factorisation step, von zur Gathen and Shoup reduced this cost to $O(d^{2+o(1)} + d^{1+o(1)}\log q)$ operations \cite{vS1992}. Kaltofen and Lobo subsequently matched this complexity by adapting Berlekamp's algorithm to a black-box linear algebra model \cite{KL1994}. The quadratic barrier was breached by Kaltofen and Shoup, who utilised fast matrix multiplication to achieve a complexity of $O(d^{1.815}\log q)$ \cite{KS1998}. More recently, Kedlaya and Umans leveraged a novel approach for modular composition to further reduce the complexity to $O(d^{1.5 + o(1)} + d^{1+o(1)}\log q)$ operations in $\mathbb{F}_q$ \cite{KU2011}.

\section{Model of computation}
\label{sec:comp-model}

We will mostly focus on \emph{arithmetic/algebraic circuits}, a model very natural for computing multivariate polynomials in the variables ${\vecx \coloneqq (x_1,\ldots,x_n)}$\footnote{We will use bold letters to denote tuples of variables.} over a field $\F$. Introduced by Valiant~\cite{Val1979} to develop the algebraic analogue of $\NP$-completeness~\cite{Val1982}, it has led to the development of a rich and varied theory of algebraic complexity (see~\cite{Bur2024}). 

\begin{figure}
    \centering
    \begin{tikzpicture}
    \begin{scope}[scale=0.50]
        \node[rectangle,scale=0.8] (f1) at (4.0, -1) {$f \in \mathbb{F}[x_1,\ldots,x_n]$};
        \node[circle,fill=Thistle1] (g6) at (4.0, -3) {$+$} edge[->] (4.0,-1.5);
        \draw[dashed] (1.0,-3) -- (g6);
        \draw[dashed] (g6) -- (8.0,-3);
        \node[rectangle] (s1) at (9.0, -3) {$\sum$};
        \node[circle,fill=PeachPuff1] (g4) at (3.0, -5) {$\times$} edge[->] (g6);
        \node[circle,fill=PeachPuff1] (g5) at (5.0, -5) {$\times$} edge[->] (g6);
        \draw[dashed] (1.0,-5) -- (g4);
        \draw[dashed] (g4) -- (g5);
        \draw[dashed] (g5) -- (8.0,-5);
        \node[rectangle] (s2) at (9.0, -5) {$\prod$};
        \node[circle,fill=Thistle1] (g1) at (2.0, -7) {$+$} edge[->] node[left] {} (g4) edge[->] (g5);
        \node[circle,fill=Thistle1] (g2) at (4.0, -7) {$+$} edge[->] (g4);
        \node[circle,fill=Thistle1] (g3) at (6.0, -7) {$+$}  edge[->] (g5);
        \draw[dashed] (1.0,-7) -- (g1);
        \draw[dashed] (g1) -- (g2);
        \draw[dashed] (g2) -- (g3);
        \draw[dashed] (g3) -- (8.0,-7);
        \node[rectangle] (s2) at (9.0, -7) {$\sum$};
        \node[circle,thick] (x1) at (1,-9) {$x_1$} edge[->] (g2) edge[->] (g1);
        \node[circle,thick] (x2) at (3.0, -9) {$x_2$} edge[->] (g2);
        \node[circle,thick] (x3) at (5.0, -9) {$\ldots$} edge[->] (g1) edge[->] (g3);
        \node[circle,thick] (x4) at (7.0, -9) {$x_n$} edge[->] (g3);
        \node[rectangle, scale=0.8] (c1) at (4.0, -10) {Circuit};
    \end{scope}

    \tikzstyle{gate}=[circle,draw=black!40,thick]
    \tikzstyle{leaf}=[circle,thick,inner sep=0]
    \begin{scope}[xshift=5cm, scale=0.50]
        \node[rectangle,scale=0.8] (f1) at (4.0, -1) {$f \in \mathbb{F}[x_1,\ldots,x_n]$};
        \node[circle,fill=Thistle1] (g6) at (4.0, -3) {$+$} edge[->] (4.0,-1.5);
        \draw[dashed] (0,-3) -- (g6);
        \draw[dashed] (g6) -- (7.0,-3);
        \node[circle,fill=PeachPuff1] (g4) at (3.0, -5) {$\times$} edge[->] (g6);
        \node[circle,fill=PeachPuff1] (g5) at (5.0, -5) {$\times$} edge[->] (g6);
        \draw[dashed] (0,-5) -- (g4);
        \draw[dashed] (g4) -- (g5);
        \draw[dashed] (g5) -- (7.0,-5);
        \node[circle,fill=Thistle1] (g1) at (2.0, -7) {$+$} edge[->] node[left] {} (g4);
        \node[circle,fill=Thistle1] (g2) at (4.0, -7) {$+$} edge[->] (g4);
        \node[circle,fill=Thistle1] (g3) at (6.0, -7) {$+$}  edge[->] (g5);
        \draw[dashed] (0,-7) -- (g1);
        \draw[dashed] (g1) -- (g2);
        \draw[dashed] (g2) -- (g3);
        \draw[dashed] (g3) -- (7.0,-7);
        \node[circle,thick] (x1) at (1,-9) {$x_1$} edge[->] (g1);
        \node[circle,thick] (x1) at (5,-9) {$x_1$} edge[->] (g2);
        \node[circle,thick] (x2) at (3.0, -9) {$x_2$} edge[->] (g2);
        \node[circle,thick] (x3) at (6.0, -9) {$\ldots$};
        \node[circle,thick] (x4) at (7.0, -9) {$x_n$} edge[->] (g3);
        \node[rectangle, scale=0.8] (c1) at (4.0, -10) {Formula};
    \end{scope}
    \end{tikzpicture}
    
    \caption{An algebraic circuit and a formula of depth $3$.}
    \label{fig:alg-ckt-formula}
\end{figure}
    
\begin{definition}[Algebraic Circuits and Formulas] \label{def:circuits}    
    An \emph{algebraic circuit}, defined over a field $\F$, is a \emph{layered directed acyclic graph} with alternating layers of `$+$' and `$\times$' gates, and a single root, called the `output' gate. The `input' \emph{leaf} gates are labelled by either a variable from $x_1, \dots, x_n$ or a constant from $\F$. If the graph is a \emph{tree}, then we call it a \emph{formula}. See \Cref{fig:alg-ckt-formula} for an illustration.
    
    A circuit computes a polynomial $f \in \F[\vecx]$ in the natural way: a `$+$' gate sums up the polynomials from its children, whereas a `$\times$' gate computes their product, with the root finally computing $f$. The \emph{size} of a circuit is the total number of vertices in the graph. The \emph{depth} of the circuit is the number of layers in the circuit, or equivalently, the length of the longest path from the root to a leaf. 
\end{definition}

A \emph{complexity class} in the algebraic setting comprises families (or sequences) of polynomials, where a family $(f_n)_{n \in \N}$, contains multivariate polynomials $f_n$ over a field $\F$, and the number of variables in $f_n$ grows polynomially with $n$.

In Boolean complexity, the notion of efficient computation is captured by the class $\P$ of problems solvable in polynomial time. The algebraic analogue over a field $\F$ is $\VP_{\F}$ (for Valiant's $\P$, called $p$-computable by Valiant) and consists of all polynomial families where $f_n$ has degree $\poly(n)$ and the smallest circuit (over $\F$) computing $f_n$ has size $\poly(n)$ (\Cref{def:vp}). We will usually drop the field from the notation when the context makes it clear. Note that the notion of computation is \emph{non-uniform} --- the circuits of $f_n$ for different $n$ need not be related to one another. A prime example of a polynomial family in $\VP$ is $(\Det_n)$, defined by the \emph{determinant} of the $n \times n$ symbolic matrix $(x_{ij})_{1\leq i,j \leq n}$:

$$\Det_n = \sum_{\sigma \in S_n} \left( \sgn(\sigma)\prod_{i=1}^n x_{i,\sigma(i)} \right).$$

The algebraic analogue of the class $\NP$ is called $\VNP$. Informally\footnote{For the formal version, see \Cref{def:vnp}.}, it consists of polynomial families which are `explicit', in the sense that given a monomial of $f_n$, we can compute the corresponding coefficient efficiently, say in polynomial time. It is not hard to show that $\VP \subseteq \VNP$, and the long-standing conjecture of Valiant~\cite{Val1979} is that there are explicit polynomial families that cannot be computed efficiently, i.e., $\VP \subsetneq \VNP$. A prominent `explicit' candidate for this separation is the family of \emph{permanents}, 

$$\Perm_n = \sum_{\sigma \in S_n} \left(\prod_{i=1}^n x_{i,\sigma(i)} \right).$$

The determinant and permanent families essentially characterise the classes $\VP$ and $\VNP$, respectively. Hence, Valiant's conjecture is also sometimes called the Permanent versus Determinant problem~\cite{Agr2006}. It is the algebraic version of Cook's hypothesis~\cite{Coo1971}, the famous $\P$ vs. $\NP$ problem (see~\cite{AB2009} for more details). There is a formal sense in which the $\VP$ vs. $\VNP$ problem is a `stepping stone' towards the $\P$ vs. $\NP$ problem~\cite{Bur2000a}. For details on the connection between Valiant's and Cook's hypotheses, and the progress on Valiant's conjecture, we encourage readers to consult~\cite{BCS1997,Bur1999,Bur2000,SY2010,CKW2010,Mah2014,Sap2021}.

\subsection{Structural results}
\label{subsec:struct-results}

Algebraic circuits impose a combinatorial structure on the polynomials being computed. We will now list (without proofs) some structural properties of algebraic circuits that showcase the robustness of the model and will be useful for us in the future. For proofs, see~\cite{SY2010,Sap2021}. To begin with, we can extract coefficients (with respect to a single variable) of a polynomial computed by a small circuit efficiently.

\begin{lemma}[Interpolation]
\label{lem:interpolation}
    Let $\F$ be a field with $|\F| > k$, and $f \in \F[\vecx,y]$ be a polynomial with $\deg_y(f) = k$. Suppose that $f(\vecx,y)=\sum_{j=0}^k f_j(\vecx)y^j$ where $f_j \in \F[\vecx]$ for all $j \in \{0,1,\ldots,k\}$. 
    
    If $f(\vecx,y)$ can be computed by a circuit of size $s$ and depth $\Delta$, then for all $j \in \{0,1,\ldots,k\}$, $f_j(\vecx)$ can be computed by a circuit of size $O(sk)$ and depth $\Delta +1$.
\end{lemma}

As we are concerned only with multivariate \emph{polynomials}, we can always formally define the partial derivative of a polynomial with respect to a variable (and by extension, multiple variables) over any field $\F$. Computing the partial derivatives of a circuit with respect to a variable of bounded individual degree is an efficient operation.

\begin{lemma}[Partial Derivatives]
\label{lem:partial-derivative}
    Let $\F$ be a field with $|\F| > r$, and $f \in \F[\vecx,y]$ be a polynomial with $\deg_y(f)=r$. If $f(\vecx,y)$ can be computed by a circuit of size $s$ and depth $\Delta$, then for all $0 \leq j \leq r$, the partial derivative $\partial_y^j f(\vecx,y)$ can be computed by a circuit of size $O(sr^3)$ and depth $\Delta$. 
\end{lemma} 

The observant reader will have noticed the conspicuous absence of divisions in our definition of algebraic circuits. It turns out that our notions of complexity do not depend on this exclusion, for the most part. Division can be eliminated efficiently over any field~\cite{Str1973a, HY2011}. 

\begin{lemma}[Division Elimination] \label{lem:division-elimination}
   Let $f \in \F[\vecx]$ be a polynomial computed by a circuit (with division gates) of size $s$.
   Then for every $d \in \N$, there exists a circuit without division gates of size $\poly(s, d)$ that computes $f \bmod \inangle{\vecx}^{d+1}$.
\end{lemma}

In the univariate factoring algorithm, we saw that $\gcd$ computation was an essential ingredient. The following lemma shows that it can be computed efficiently in the algebraic circuit model. Refer to \cite[Lemma 2.9]{KSS2015} for the complete proof.

\begin{lemma}[GCD in circuits] \label{lem:gcd-ckt}
    Let $f,g \in \F[\vecx]$ be two $n$-variate polynomials of degree at most $d$ computed by circuits of size $s$. Then, the $\gcd(f,g)$ can be computed by a circuit of size $\poly(s, d, n)$.
\end{lemma}

\section{Applications of factoring}
\label{sec:factor-appl}

Polynomial factoring is not only a great mathematical problem. The techniques developed for its solution have a wide range of applications in various areas of computer science. We briefly describe some of them.

\subsection{Equivalence to identity testing}
\label{sec:factor-identity}

Consider a class of polynomials $\calC$ that can be computed by algebraic circuits. The \emph{polynomial identity testing} (PIT) problem for $\calC$ asks whether a given polynomial $f \in \calC$ is identically zero. Due to the PIT lemma~\cite{ore1922hohere,DL1978,Zip1979,,Sch1980}, a simple randomised algorithm for this problem has been known for decades. However, designing a polynomial-time deterministic algorithm for PIT for the class $\VP$ remains a long-standing open problem. Nevertheless, several efficient deterministic algorithms have been developed for restricted circuit classes. The problem is studied in two settings: the \emph{black-box} setting, where only polynomial evaluation is permissible, and the \emph{white-box} setting, where the internals of the circuit are accessible. 

The importance of this fundamental problem stems from its applications to equivalence checking, perfect matching, primality testing, and more. The assumption of efficient identity testing algorithms has led to a variety of derandomisation results. Most notably, PIT is tightly connected to proving strong lower bounds for algebraic circuits (more in \Cref{sec:hardness-randomness}). For a comprehensive treatment, we refer the reader to classical surveys such as \cite{SY2010, Sax2009, Sax2014}, as well as the recent exposition in \cite{DG2024} and the references therein.

PIT was first linked to factorisation when Shpilka and Volkovich~\cite{SV2010} observed that the polynomial $f(\vecx) + yz$ has two irreducible factors over disjoint sets of variables if and only if $f(\vecx)$ is identically zero. This observation implies that a deterministic algorithm for multivariate polynomial factorisation would suffice to derandomise PIT.

The connection between the two problems was further solidified by the work of Kopparty, Saraf, and Shpilka~\cite{KSS2015}, who showed that a derandomised PIT algorithm also leads to a deterministic multivariate factoring algorithm. Together, these results establish the equivalence between derandomising PIT and polynomial factorisation in both black-box and white-box settings.

A natural and simpler question related to polynomial factoring is that of \emph{divisibility testing}: given two polynomials $f$ and $g$, determine whether $g$ divides $f$. One could factor both $f$ and $g$ followed by comparing the irreducible factors using PIT to solve this problem using randomisation. However, since non-trivial deterministic PIT algorithms are known in several restricted settings, it is natural to ask whether divisibility testing is easier in such settings.

This question was first studied by Saha, Saptharishi, and Saxena~\cite{SSS2013}, who reduced the problem of testing divisibility of sparse polynomials by a linear polynomial to PIT for expressions of the form $\sum_{i} \vecx^{\veca_i} \cdot f_i(\vecx)^{d_i}$, where each $\vecx^{\veca_i}$ is a monomial and $\deg f_i \leq 1$. Known PIT algorithms for such structured polynomials were applied to obtain efficient and deterministic divisibility testing algorithms (see~\cite{RS2005, FS2013}). Later, Forbes~\cite{For2015} extended this line of work to test divisibility of sparse polynomials by constant-degree polynomials. 
The problem in this case was reduced to PIT for almost similar polynomials of the form $\sum_{i} \vecx^{\veca_i} \cdot f_i(\vecx)^{d_i}$ with $\deg f_i \leq 2$, and gave a quasi-polynomial time algorithm for this case. For certain classes of polynomials, Forbes also showed a general reduction to PIT \cite[Section 7]{For2015}. For a more comprehensive discussion on these equivalences, we refer the reader to the survey by Shpilka and Forbes~\cite[Section 5]{FS2015}.

\subsection{Hardness vs Randomness} \label{sec:hardness-randomness}

In the previous section, we hinted at the tight connection between the derandomisation of PIT and strong lower bounds for algebraic circuits. Informally, this connection implies that PIT can be efficiently derandomised if and only if there exist explicit polynomials of high circuit complexity. Interestingly, factoring results on the algebraic circuit classes play a surprising yet pivotal role in establishing this connection.

Heintz and Schnorr~\cite{HS1980}, and subsequently Agrawal~\cite{Agr2005}, proved that efficient black-box identity testing implies strong circuit lower bounds. Specifically, they showed that a deterministic polynomial-time black-box PIT algorithm for circuits of size $s$ guarantees the existence of a polynomial with coefficients in $\mathsf{PSPACE}$ that requires algebraic circuits of size exponential in $s$.
This yields an exponential lower bound from a polynomial-time PIT algorithm. A central open question is whether the complexity of computing the coefficients can be reduced to $\sPbyPoly$, as this would imply a separation of $\VP$ from $\VNP$~\cite[Open Problem 17]{SY2010}. Kabanets and Impagliazzo~\cite{KI2004} further strengthened this direction of the connection by proving that derandomising PIT, even in the weaker white-box setting, would imply either $\VP \neq \VNP$ or $\mathsf{NEXP} \not\subset \P/\poly$.

Kabanets and Impagliazzo strengthened this connection by proving the reverse direction, drawing inspiration from Nisan and Wigderson's results in the Boolean setting~\cite{NW1994}. To describe this, we introduce an alternative notion of black-box identity testing. Consider a polynomial map $\calG \coloneqq (g_1, \dots, g_n): \F^r \to \F^n$ with seed length $r < n$. The map $\calG$ is a \emph{hitting-set generator} for a class of polynomials $\calC$ if, for every non-zero $f \in \calC$, the composed polynomial $f \circ \calG = f(g_1(\vecy), \dots, g_n(\vecy))$ is also non-zero. For a detailed discussion of the connections between hitting-set generators and PIT, see~\cite[Section 3.2.2]{For2014}.
The high-level idea of Kabanets and Impagliazzo is as follows: to test whether an $n$-variate polynomial $f \in \calC$ is identically zero, construct a hitting-set generator $\calG$ of seed length $r = O(1)$, where each $g_i$ has degree $\poly(n)$. Identity testing of $f$ then reduces to checking whether the constant-variate (degree-$r$) polynomial $f \circ \calG$ is identically zero---a problem efficiently solvable via the PIT lemma.

To illustrate the idea of a hitting-set generator in action, consider a generator with seed length $(n-1)$ defined as $\calG(\vecy) = (y_1, \dots, y_{n-1}, g(\vecy))$, where $g$ is a hard polynomial. This generator extends $r = n-1$ variables by one additional coordinate using a hard polynomial. Now, suppose $f(\vecy, x)$ is a non-zero polynomial in $\calC$. For the sake of contradiction, assume that $f \circ \calG = f(y_1, \dots, y_{n-1}, g(\vecy)) = 0$. This implies that $(x - g)$ divides $f$. 
As we will discuss later (\Cref{subsec:small-circuits}), if $f$ is computable by a small algebraic circuit, then so are all of its factors.
This phenomenon is known as \emph{closure under factoring}.
Therefore, $(x - g)$ must also be computable by a small circuit, contradicting our assumption that $g$ has large circuit complexity. Hence, $\calG$ is as a valid hitting-set generator for $\calC$.

Moving beyond the single-variable stretch, Kabanets and Impagliazzo~\cite{KI2004} adapted the Nisan-Wigderson design~\cite{NW1994} to the algebraic setting.
Their construction yields a generator $\calG_{\mathsf{KI}}$ that stretches a seed of length $r = \poly(\log n)$ to $n$ variables, using a polynomial $g$ that is assumed to be hard for arithmetic circuits.
Consequently, this generator leads to a black-box PIT algorithm.
Crucially, the correctness of this derandomisation relies fundamentally on the closure of algebraic circuits under factoring, which will be explored in detail in the upcoming parts of this survey.

\begin{theorem}[Combinatorial Hardness Implies PIT] \label{thm:hardness-pit-combinatorial}
    Suppose there exists a family $\{f_n\}_{n \in \mathbb{N}}$ of $n$-variate multilinear polynomials such that each $f_n$ requires algebraic circuits of size $2^{\Omega(n)}$. 
    Then, there exists a black-box PIT algorithm running in time $2^{\poly(\log n)}$ for the class of $n$-variate polynomials of degree $\poly(n)$ computable by circuits of size $\poly(n)$.
\end{theorem}

The core argument in \cite{KI2004} proceeds by contrapositive: if a non-zero polynomial $f$ vanishes on the image of the generator (i.e., $f \circ \calG_{\mathsf{KI}} = 0$), then the hard polynomial $g$ used to construct $\calG_{\mathsf{KI}}$ must essentially be computable by a small circuit, contradicting the hardness assumption.

However, over fields of characteristic $p > 0$, this factorisation argument encounters a significant obstacle.
The original proof could only deduce that some power $g^{p^k}$ (where $p^k \leq \deg(f)$) has a small circuit, rather than $g$ itself.
This gap remained open until Andrews~\cite{And2020} showed that $p$-th roots of circuits can be computed efficiently, provided the number of variables is small.
Andrews' result thus allows for the efficient recovery of $g$ from $g^{p^k}$, extending the hardness-to-randomness tradeoffs to fields of low characteristic.

\Cref{thm:hardness-pit-combinatorial} was subsequently improved with better parameters by the authors of \cite{GKSS2022}. They showed that a constant-variate hard polynomial can imply a polynomial-time black-box PIT algorithm. A key contribution in their work is the use of an \emph{algebraic} generator, in contrast to the \emph{combinatorial design}-based generator used in~\cite{KI2004}. The combinatorial design inherently suffers from limitations---most notably, it cannot yield better than a quasi-polynomial time PIT algorithm.

\begin{theorem}[Non-combinatorial Hardness-to-PIT] \label{thm:hardness-pit}
    Let $k \coloneqq O(1)$. Suppose $\{ f_d \}_{d \in \N}$ is a family of $k$-variate polynomials of degree $d$, and each $f_d$ requires algebraic circuits of size at least $d^{0.1}$. Then, there exists a black-box PIT algorithm running in time $s^{O(k^2)}$ for the class of $s$-variate, degree-$s$ polynomials computable by algebraic circuits of size $s$.
\end{theorem}

The main limitation of the above results in extending them to restricted models such as formulas and constant-depth circuits was the lack of closure under factoring for these models. Only weaker statements were known earlier~\cite{DSY2009,CKS2019,And2020}. The recent resolution of these closure questions for restricted models by~\cite{BKRRSS2025} proved that the hardness-versus-randomness connection holds for these models as well.

\begin{theorem}[\protect{\cite[Theorem 6.1]{BKRRSS2025}}]
    Let $\Delta > 0$. Suppose $\{ f_d \}_{d \in \N}$ is a family of $n$-variate polynomials of degree $d$, and each $f_d$ requires depth-$\Delta$ algebraic circuits of size at least $2^{\Omega(n)}$. Then, there exists a black-box PIT algorithm running in time $n^{O(\log n)}$ for the class of circuits of depth $\Delta - O(1)$.
\end{theorem}

We conclude this section by highlighting an important application of the hardness versus randomness paradigm---namely, the phenomenon of \emph{bootstrapping} \cite{AGS2019, KST2023, GKSS2022}. This idea relies on recursively leveraging the connection between hardness and polynomial identity testing (PIT). Remarkably, it was proved in \cite[Theorem 1.6]{GKSS2022} that a complete derandomisation of PIT can be achieved from even a mildly non-trivial derandomisation of PIT.

\begin{theorem}[Bootstrapping] \label{thm:bootstrapping}
    Suppose there exists a black-box PIT algorithm that runs in time $\left(s^k - 1\right)$ for the class of $k$-variate polynomials of individual degree $s$, computable by algebraic circuits of size $s^{0.1}$. If $s$ is sufficiently large, then there exists a black-box PIT algorithm running in time $s^{O(k^2)}$ for $s$-variate polynomials computable by algebraic circuits of size $s$.
\end{theorem}

To prove this theorem, one first uses the assumed mildly non-trivial black-box PIT to derive an explicit hard polynomial, via the hardness-from-derandomisation results of Heintz-Schnorr~\cite{HS1980} and Agrawal~\cite{Agr2005}. One then feeds this hardness into the combinatorial hardness-to-PIT \Cref{thm:hardness-pit-combinatorial} to obtain the stronger PIT algorithm claimed.
A technical point is that this argument requires carefully matching parameters across several steps which we omit here.
For a more detailed exposition of this connection between factoring and its implications for the hardness versus randomness frontier, we refer the reader to the survey by Kumar and Saptharishi~\cite{KS2019}.

\subsection{Other applications}
\label{subsec:misc-appl}

The field of \emph{coding theory}~\cite{GRS2023} deals with developing \emph{error-correcting codes} --- ways of adding (minimal) redundancy to data such that even if parts of it get corrupted during transmission, one can recover the original information. Reed--Solomon codes are particularly ubiquitous and also `optimal', in a sense. They treat the original message as a univariate polynomial and the encoding is the evaluation of this polynomial at various points over some finite field. When the number of errors is too large, we cannot decode a corrupted message uniquely, but we can produce a small list of potential decodings (also known as \emph{list decoding}). The list decoding algorithm of Sudan~\cite{Sud1997} and the later improvement by Guruswami and Sudan~\cite{GS1999} crucially use polynomial factorisation. We point the reader to the survey of Forbes and Shpilka~\cite[Section 3.1]{FS2015} for more details.

The problem of learning algebraic circuits is called \emph{reconstruction}~\cite[Chapter 5]{SY2010}. We are given black-box access to a polynomial computed by a circuit $C$ from some nice family of circuits $\mathcal{C}$, and we need to `learn' an arithmetic circuit computing the same polynomial as $C$. Efficient polynomial factorisation plays an important role in many reconstruction algorithms~\cite{KS2009,Sin2016,Sin2022,SS2025}. Polynomial factorisation is also helpful in algebraic property testing~\cite{AS2003} and the construction of pseudorandom generators for low-degree polynomials~\cite{Bog2005,DGV2024}.  

In \emph{proof complexity}, a central problem is to prove that certain propositional tautologies need extremely lengthy proofs, even in very powerful proof systems. An important work of Cook and Reckhow~\cite{CR1979} showed that such proof complexity lower bounds, provided we are able to show them for every propositional proof system, would separate the complexity classes $\NP$ and $\coNP$, and in turn, also $\P$ from $\NP$. Closure of a class under factoring is another way of saying that multiples of hard polynomials from the class remain hard (more generally, one can study the complexity of ideals~\cite{Gro2020}). Using such hard multiples, Forbes, Shpilka, Tzameret, and Wigderson~\cite{FSTW2021} showed lower bounds against certain algebraic proof systems. For more about proof complexity, see~
\cite{Kra1995,Juk2012,Kra2019,AF2022,HLT24,BLRS2025,LuST26}.

Polynomial factorisation also has applications to various other problems in mathematics, such as derandomising Noether's Normalisation Lemma~\cite{Mul2017}, the primary decomposition of polynomial ideals~\cite{GTZ1988}, and isomorphism of algebras~\cite{KS2006a,IKRS2012}. 
In cryptography, polynomial factorisation is often used as a subroutine, for example, in index calculus algorithms and public-key encryption \cite[Chapter 3]{MvV1997}, and in factoring integers \cite{Bre2000}.
Polynomial factorisation algorithms, and the tools developed alongside them, have proven invaluable in the cryptanalysis of lattice-based schemes~\cite{NV2010} as well as post-quantum cryptosystems~\cite{DPS2020}.

\section {Factoring via Hensel lifting}
\label{sec:hensel-lifting}

Hensel lifting was first introduced by Kurt Hensel in a series of papers \cite{hensel1897, Hen1904, hensel1908theorie, Hen1918}, although an earlier form of it seems to have been known to Gau\ss ~\cite{Fre2007}. For an element $p$ in a ring $\calR$, Hensel lifting gives a method to compute factorisation modulo $p^\ell$ (for any $\ell > 0$) from the factorisation modulo $p$.
The term \emph{lifting} refers to the process of improving the \emph{precision} of the factor, by iteratively refining a coarse factorisation modulo $p$ to modulo higher powers $p^\ell$.

\begin{theorem}[Hensel lifting] \label{thm:hensel}
    Let $\ringR$ be a ring and $\calI$ be an ideal in $\ringR$. Consider elements $f,g, h \in \ringR$ such that $f \equiv g \cdot h \pmod{\calI}$ and there exist $u,v \in \ringR$ such that $u \cdot g + v \cdot h \equiv 1 \pmod{\calI}$. Then, we have:
    \begin{enumerate}
        \item{\bf Existence.} There exist $g', h' \in \ringR$ such that $f \equiv g' \cdot h' \pmod{\calI^2}$ and
        \begin{align*}
            g' \;\equiv\; g \pmod{\calI} \quad \text{and} \quad h' \;\equiv\; h \pmod{\calI}
        \end{align*}
        \item{\bf Pseudo-Coprimality.} For some $u' \equiv u \pmod{\calI}$ and $v' \equiv v \pmod{\calI}$, we have $u' \cdot g' + v' \cdot h' \equiv 1 \pmod{\calI^2}$.
        \item{\bf Uniqueness.} If any other $\tilde{g}, \tilde{h}$ satisfy the above conditions, then there must be a $\mu \in \calI$ such that $\tilde{g} \equiv g'(1+\mu) \pmod{\calI^2}$ and $\tilde{h} \equiv h'(1-\mu) \pmod{\calI^2}$.
    \end{enumerate}
\end{theorem}

The polynomials $g'$ and $h'$ defined in the Hensel lifting theorem above are called the \emph{lifts} of $g$ and $h$, respectively. 
The theorem is not merely existential; it is constructive and provides an explicit procedure to compute the lifts. 
Let $e \coloneqq f - gh$. Then the lifts can be obtained as follows:
\begin{equation}
    g' \;\coloneqq \; g + e \cdot v \quad \text{and} \quad h' \; \coloneqq \; h + e \cdot u. \label{eq:hensel-lift}
\end{equation}

To apply Hensel lifting for polynomial factoring, we typically set the ring $\ringR = \F[x, y]$ and choose the ideal $\calI = \langle y^k \rangle$ for some $k \geq 1$. 
In the case of a multivariate polynomial, we transform it to a bivariate polynomial. We will discuss this in more detail in the upcoming sections.
When working over a ring of polynomials, a monic version of the theorem further guarantees the uniqueness of the monic lifts.

\begin{theorem}[Monic Hensel lifting] \label{thm:monic-hensel}
    Let $\F[x, y]$ be a ring of polynomials. Consider monic polynomials $f,g, h \in \F[x,y]$ such that $f \equiv g \cdot h \pmod{y}$ and there exist $u,v \in \F[x, y]$ such that $u \cdot g + v \cdot h \equiv 1 \pmod{y}$. Then, we have:
    \begin{enumerate}
        \item{\bf Existence.} There exist monic polynomials $g', h' \in \F[x,y]$ such that $f \equiv g' \cdot h' \pmod{y^2}$ and
        \begin{align*}
            g' \;\equiv\; g \pmod{y} \quad \text{and} \quad h' \;\equiv\; h \pmod{y}
        \end{align*}
        \item{\bf Pseudo-Coprimality.} For some $u' \equiv u \pmod{y}$ and $v' \equiv v \pmod{y}$, we have $u' \cdot g' + v' \cdot h' \equiv 1 \pmod{y^2}$.
        \item{\bf Uniqueness.} If any other monic $\tilde{g}, \tilde{h}$ satisfy the above conditions, then $\tilde{g} \equiv g' \pmod{y^2}$ and $\tilde{h} \equiv h' \pmod{y^2}$.
    \end{enumerate}
\end{theorem}

Once again let $e \coloneqq f - gh$, and define $\hat{g}, \hat{h}$ as in \Cref{eq:hensel-lift}. Compute the following expressions using the division with remainder algorithm:
$$ \frac{\hat{g} - g}{y} \; = \; q \cdot g + r .$$
Then the unique monic lifts can be computed as follows:
\begin{equation}
    g' \;\coloneqq\; g + y \cdot r \quad \text{and} \quad h' \;\coloneqq\; \hat{h} \cdot (1 + q \cdot y) \label{eq:monic-hensel-lift}
\end{equation}
An expository proof of Hensel lifting can be found in \cite[Lemma 3.4]{KSS2015}.

\subsection{Polynomials in the dense representation} \label{sec:poly-factoring-dense}

To demonstrate the use of Hensel lifting in factoring, consider the problem of factoring a bivariate polynomial $f \in \mathbb{F}_q[x,y]$ of degree at most $d$, where $q = p^a$. The input is given in the dense representation, as in \Cref{sec:univariate-factor}. In fact, the univariate factoring algorithms will be used as subroutines in the bivariate setting. 
For convenience, we view $f$ as an element of $(\mathbb{F}_q[x])[y]$, effectively shifting $x$ into the base ring and treating it as a constant.

\paragraph{\bf Resultants.} 
We saw in \Cref{sec:univariate-factor} that GCD is an important tool for univariate factoring algorithms.
A closely related polynomial called the \emph{resultant} plays an important role in several factoring algorithms.

\begin{definition}[Resultant] \label{def:resultant}
    Consider two $n$-variate polynomials $f, g \in \F[\vecy][x]$ as follows:
    $$f(\vecy, x) \;=\; \sum_{i = 0}^{d_1} f_i(\vecy) \cdot x^i \quad \text{and} \quad g(\vecy, x) \;=\; \sum_{i=0}^{d_2} g_i(\vecy) \cdot x^i.$$ Define the \emph{Sylvester matrix} of $f$ and $g$ as the following $(d_1 + d_2) \times (d_1 + d_2)$ matrix:
    \begin{align*}
        \mathbf{S}_x(f, g)\; =\;
        \begin{pmatrix}
            f_0 &   &   &  & g_0 &  &  &  &  \\
            f_1 & f_0 &   &  & g_1 & g_0 &  &  &  \\
            f_2 & f_1 & \ddots &  & \vdots & g_1 & \ddots &  &  \\
            \vdots & \vdots & \ddots &   & g_{d_2} & \vdots & \ddots & \ddots &  \\
            f_{d_1} & f_{d_1-1} &   & f_0 &   & g_{d_2} &   & \ddots & g_0 \\
             & f_{d_1} & \ddots & f_1 &  &  & \ddots &   & g_1 \\
             &  & \ddots & \vdots &  &  & \ddots & \ddots & \vdots \\
             &  &  & f_{d_1} &  &  &  & \ddots & g_{d_2}
        \end{pmatrix}.
    \end{align*}
    Then, the resultant of the two polynomials with respect to $x$ is defined as the determinant of the Sylvester matrix: $$\Res_x(f,g) \;\coloneqq \;\Det(\mathbf{S}_x(f, g)).$$
\end{definition}

\noindent The resultant is a polynomial of degree at most $d_1 + d_2$ in the coefficients of $f$ and $g$.
The following lemma captures one of several properties of the resultant that are useful in factoring algorithms. The second part of the lemma can be proved using the first part and Cramer's rule. See \cite[Section 6.3]{vG2013} for a detailed proof.

\begin{lemma}[Resultant and GCD] \label{lem:res-gcd}
    Let $f, g \in \F[x][y]$ be polynomials of positive degree in $x$. Then the following is true:
    \begin{enumerate}
        \item $\Res_x(f,g) = 0$ if and only if $f$ and $g$ share a common factor with positive degree in $x$.
        \item There exist polynomials $u, v$ such that $u \cdot f + v \cdot g = \Res_x(f,g)$.
    \end{enumerate}
\end{lemma}

\paragraph{\bf Bivariate Factoring.} 
The high-level idea of bivariate factoring is to first obtain factorisation of $f$ modulo $y$, and then lift them using Hensel lifting. But before we can do that, we need to ensure that $f$ satisfies certain properties, as described below.

\begin{enumerate}
    \item{\bf$f(x, 0)$ is square-free:} For $a \in \mathbb{F}_q$, the univariate polynomial $f(x, a)$ is square-free if and only if $$\gcd\bigg(f(x,a), \partial_x f(x, a)\bigg) \;\neq\; 1.$$ From \Cref{lem:res-gcd}, this is equivalent to $\Res_x(f, \partial_x f) \rvert_{y = a} \,\neq\, 0$.
    Since the degree of the resultant in $y$ is at most $2d^2$, it suffices to test $O(2d^2)$ values of $a \in \mathbb{F}_q$ (or a suitable extension) to find one for which $f(x, a)$ is square-free.
    A simple linear transformation will ensure that the resultant is non-zero at $a = 0$.
    Further, this would also imply that $f$ is square-free.
    \item \textbf{$f$ is monic in $x$}: Let $f_d(y)$ be the leading coefficient of $f$ of degree at most $d$. Then, $$\tilde{f} \;=\; f_d^{d-1} \cdot f\bigg(x / f_d, y\bigg)$$ is a monic polynomial in $x$.
\end{enumerate}

For simplicity, let $f$ denote the monic and square-free polynomial obtained from the transformations above.
Ensuring that $f(x, 0)$ is square-free is crucial, as it allows us to apply univariate factoring algorithms (\Cref{sec:univariate-factor}) to find $g_0, h_0 \in \F[x]$ such that
\[
    f \equiv g_0 \cdot h_0 \pmod{y},
\]
where $g_0$ is a monic irreducible factor in $x$.
We then apply Hensel lifting (\Cref{thm:monic-hensel}) for $t$ iterations to obtain $g_t$ and $h_t$ satisfying
\[
    f \;\equiv\; g_t \cdot h_t \pmod{y^{2^t}}.
\]
To understand the nature of $g_0$, consider the irreducible factorisation $f = f_1, \dots, f_r$.
When restricted to $y=0$, a bivariate factor $f_i$ may become reducible.
Consequently, $g_0$ corresponds to an irreducible factor of the restriction of some $f_i(x, 0)$.
Further, the lifting process continues until $2^t > 2d^2$, a bound chosen to exploit the properties of resultants, as we address in the proof of the following claim.

\begin{proposition} \label{prop:factor-existence}
    If the input polynomial $f$ is reducible, then there exists a non-trivial factor $f_1$ of $f$ such that $f_1 \equiv g_t \cdot \ell \pmod{y^{2^t}}$ for some polynomial $\ell$. Furthermore, $\deg_x(f_1) \le \deg_x(f)$ and $\deg_y(f_1) \le \deg_y(f)$.
\end{proposition}

\begin{proof}
    First let us prove that such a factor exists. Since $g_0$ is an irreducible factor of $f \pmod{y}$, we know that $g_0$ must divide some irreducible factor of $f \pmod{y}$, say $f_1$. Let $f = f_1 \cdot h$. Then, for some $\ell_0$, $$f_1 \equiv g_0 \cdot \ell_0 \pmod{y}.$$ By \Cref{thm:monic-hensel} we can lift $g_0$ as a factor of $f_1$ and get $f_1 \equiv g'_t \cdot \ell_t \pmod{y^{2^t}}$. Multiplying both sides by $h$ gives $$f \equiv g'_t \cdot h'_t \pmod{y^{2^t}}.$$ By the uniqueness property of \Cref{thm:monic-hensel}, it follows that $g'_t \equiv g_t \pmod{y^{2^t}}$, which implies the existence of $f_1 \equiv g_t \cdot \ell_t \pmod{y^{2^t}}$.
    
    To see that $f_1$ is non-trivial, let us assume for the sake of contradiction that $\gcd_x(f, f_1) = 1$. Then there exist $u, v$ such that $$u \cdot f + v \cdot f_1 = \Res_x(f, f_1).$$ 
   Substituting the lifted expressions, we get $$g_t \cdot (u \, h_t + v \, \ell_t) \equiv \Res_x(f, f_1) \pmod{y^{2^t}}.$$ 
   Since $g_t$ is monic, and $\Res_x(f, {f_1}) \in \F_q[y]$, it follows that $$u \, h_t + v \, \ell_t \equiv 0 \pmod{y^{2^t}}.$$ 
   Since $2^t$ is greater than the degree of the resultant $2d^2$, it follows that $\Res_x(f, {f_1}) = 0$. This yields a contradiction, and thus $\gcd_x(f, {f_1}) \neq 1$.
\end{proof}

The above proposition suggests that a candidate polynomial $\tilde{g}$ can be obtained by solving the following linear system:
\begin{equation}
    \tilde{g} \equiv g_t \cdot \ell \pmod{y^{2^t}}, \label{eq:hensel-linear-sys}
\end{equation}
where the degree bounds for $\tilde{g}$ are taken to be the same as those guaranteed by the previous claim.
Finally, we compute $\gcd(f, \tilde{g})$ to obtain the required non-trivial factor $g$ of $f$.

Finally, we recall that the algorithm factors a polynomial that results from a sequence of preprocessing steps. Hence, it is necessary to undo these steps—particularly the monic transformation—in order to recover a factor of the original input polynomial. By applying the Cantor--Zassenhaus algorithm for univariate factorisation and analysing the cost of lifting factors via Hensel lifting, it can be observed that the overall complexity of bivariate factoring is $\poly(d, \log q)$ operations over $\F_q$.

Naturally, the bivariate factoring algorithm extends to the multivariate setting. However, as the number of variables increases, the runtime of the algorithm degrades rapidly. Later, we will explore how to address this challenge while still reusing insights from the factoring algorithms discussed thus far.

\subsection{Polynomials over integers}\label{sec:factoring-int}

The finite field factoring algorithm can be extended to factoring polynomials over integers and thereby over rationals as well.
Although the details are outside the scope of this survey, we will briefly discuss the high-level idea of the algorithm.
Consider a degree $d$ input polynomial $f \in \Z[x]$ with coefficients of bit-length at most $\ell$. 
Therefore, the coefficients of $f$ are between $-2^{\ell}$ and $2^{\ell}$, and the goal is to obtain a non-trivial factor of $f$ in $\poly(d, \ell)$ time.

The algorithm for factoring polynomials over the integers closely follows the same template as the bivariate factoring algorithm. It begins by choosing a sufficiently large prime $p$ such that $f \bmod p$ is square-free. The prime $p$ plays a role analogous to the variable $y$ in the bivariate setting.
The polynomial $f$ is then factored modulo $p$ using a univariate factoring algorithm (see \Cref{sec:univariate-factor}). The resulting factors are lifted to a factor modulo $p^{2^t}$ via Hensel lifting. As in the bivariate case, it can be shown that $\log(n^3 \ell)$ iterations of the lifting step suffice.

The most challenging part of the algorithm lies in solving the linear system similar to \Cref{eq:hensel-linear-sys}, under the additional constraint that the coefficients of the factor $g$ are small. Although this is not immediately obvious, the problem reduces to finding a shortest vector in a lattice---a task that is known to be $\NP$-hard in general.

Nevertheless, the celebrated algorithm of Lenstra, Lenstra, and Lovász \cite{LLL1982} provides an efficient \emph{weak-approximation} algorithm for this problem, which turns out to be sufficient for the purposes of polynomial factoring. For a detailed exposition, see~\cite[Chapter~13]{Sap2017} and~\cite[Section 16.5]{vG2013}.

\subsection{Polynomials over \texorpdfstring{$p$}{p}-adic numbers}
\label{subsec:p-adics}

There is no {\em topology}, nor a geometric interpretation, inherent in a finite field. Classically, this motivated mathematicians to `extend' a finite field $\F_q$ ($q$ is a power of a prime $p$) to a characteristic {\em zero} field. The latter is called an {\em unramified $p$-adic field} construction. 

{\bf $p$-adics.} For simplicity, consider the object $\Z_p$, the ring of {\em $p$-adic integers}. Its elements are infinite series of the type $(a_0 + a_1\cdot p + a_2\cdot p^2 + \cdots)$ with the {\em digits} $a_i$'s in $[0\ldots p-1]$. The notion of {\em convergence} here requires an {\em ultra-}metric, defined via prime-powers, $|p^i|_u \coloneqq p^{-i}$. (Exercise: Check that it defines a {\em metric}, i.e. a map from $\Z_p$ to $\Q$.) Also, we sometimes call $i$ the {\em $p$-adic valuation of $p^i$}. This forces an {\em un}real comparison of real numbers: $1 > p > p^2 > \cdots > p^\infty = 0$ ! The convergence allows certain integers to become invertible or {\em units}, e.g.~$1/(1-p) = 1+p+p^2+\cdots \in \Z_p$. It is easy to show: (1) $\Z_p$ is an integral domain, so it has a field of fractions called {\em $p$-adic rationals} $\Q_p$, which has non-integer elements like $1/p$ and $1/(p^2+2p^3)$, and (2) The characteristic of $\Z_p$ and $\Q_p$ is $0$. 

Given a polynomial $f(x)\in \Q_p[x]$ as a binary string of size $n$, can we factor it, over $\Q_p$, in time $\poly(n)$? A randomised poly-time algorithm was provided by \cite{Chi1987, CG2000}. The two papers are written in very different styles. We will sketch the basic ideas common to them, by working through an example.

\medskip\noindent
Consider the $2$-adic polynomial $f(x) = (x-2)^2(x-4) + 32 \in \Q_2[x]$. Note that $f \equiv x^3 \pmod 2$ has no coprime factors over $\F_2$, so the algorithm of Section \ref{sec-CZ} cannot be applied to get factors, or check irreducibility, `lifted' to $\Z/4\Z$. This forces us to work modulo higher $2$-powers and to perform some new operations. Seeing broadly, there are two major algebraic themes in the algorithm, which we oversimplify for the sake of presentation as follows. 
\begin{figure}
    \centering
    \begin{tikzpicture}[scale=0.8]
        \draw[->] (-0.5,0) -- (5,0) node[right] {$x$ monomial};
        \draw[->] (0,-0.5) -- (0,5) node[above] {coeff.valuation $p=2$};

        \foreach \x/\y in {0/4, 1/2, 2/3, 3/0} {
            \filldraw[blue] (\x,\y) circle (2pt);
            \node[above right] at (\x,\y) {(\x,\y)};
        }

        \draw[thick, red] (0,4) -- (1,2) -- (3,0);
    \end{tikzpicture}
    \caption{The Newton Diagram for $f(x)$.}
\end{figure}

{\bf Theme 1 (Newton-Hensel).} Consider the {\em Newton diagram} of the example polynomial above, $f = x^3 - 8x^2 + 20x + 16$, by plotting the monomial-exponents in the X-axis (namely, $\{0,1,2,3\}$) and the $2$-adic valuation of the integral-coefficients in the Y-axis (respectively, $\{4,2,3,0\}$). The lower-boundary of this diagram has two edges of slopes $1$ and $2$ respectively, suggesting that the $\overline{\Q}_2$-roots of $f(x)$ are (exactly) divisible by $2^1$ and $2^2$ respectively. This gives us the  `transformation' $f(2x)$ to study, and to find the first factor by a version of {\em Hensel lifting}. Thus, $f(2x)/8 = (x-2)(x-1)^2 + 4$ $=\, (x-2+2^2+2^5+\cdots)\cdot$ $(x^2 - (2+2^2+2^5+\cdots)x + 1-2^4)$. The key property we use here is the coprimality of the two factors $(x-2),(x-1)^2$ mod $2$; so, the algorithm of Section \ref{sec-CZ} can be applied over the finite field $\F_2$. 

An inverse transformation ($x\mapsto x/2$) gives us the two factors of $f$ as, $g_1\coloneqq (x-2^2+2^3+2^6+\cdots)$ and $ g_2\coloneqq (x^2 - (2^2+2^3+2^6+\cdots)x + (2^2-2^6+\cdots))$, respectively. Clearly, $g_1$ is an {\em irreducible} $2$-adic factor of $f$. What about $g_2$?

{\bf Theme 2 ($p$-adic ramification).} Again from the Newton diagram of $g_2$ we learn that it has two distinct $\overline{\Q}_2$-roots, both (exactly) divisible by $2^1$. So, we study the new transformation $g_2(2x)/4 = (x^2 - (2+2^2+2^5+\cdots)x + 1-2^4) =: T(x)$. Clearly, $T \equiv (x-1)^2 \pmod 4$ with no coprime factors $\bmod 2$. In this case both the tricks of Newton and Hensel fail. 

We have $T \equiv x^2-6x+1 \equiv (x-3)^2 - 8  \pmod 16$. From here we learn two properties: (1) $T$ is irreducible (yielding a {\em certificate} of irreducibility for $g_2$), as $\sqrt{8}$ does not exist in $\Q_2$; and (2) 
$T$ suggests a more intricate transformation to progress to a deeper root of $f$, namely, $(x-3)\mapsto (x-3)\sqrt{2}$ over the ({\em ramified}) field $\Q_2(\sqrt{2})$. 

In general, if a repeated number of these steps achieve a ramification degree equal to that of the degree of $f$, then we have a certificate of irreducibility of $f$ over $\Q_p$.

\smallskip\noindent
The above two algebraic themes give a randomised poly-time algorithm to factor $f(x)\in\Q_p[x]$. The time-complexity is based on analysing the (Galois) symmetries of both the ramified and unramified field extensions of $\Q_p$ that the algorithm constructs.



\begin{questype}{Open Problem}
    Given an integral polynomial $f(x)\in \Z[x]$ of degree $d$ and a prime-power $p^k$, can we factor $f\bmod p^k$, in randomised  $\poly(dk\log p)$ time?
\end{questype}

In this case, there is no unique factorisation property, and we cannot use division by $p$, as it is now a zero-divisor. See \cite{DMS2021, CDS2024} for a detailed survey.

\subsection{Small algebraic circuits} \label{subsec:small-circuits}

The factoring algorithms discussed in earlier sections assume that the input polynomial is given in the dense representation, that is, as an explicit list of its coefficients. However, algebraic circuits introduced in \Cref{sec:comp-model} provide a much more succinct way to represent polynomials. Consequently, the number of non-zero coefficients of a polynomial computed by a small circuit can be exponential in the circuit size, making the conversion to dense representation inefficient. Therefore, we require algorithms that can operate directly on the circuit model. We begin by formally describing the model of computation for the algorithm.

Let $f \in \F[\vecx]$ be an $n$-variate polynomial computed by an algebraic circuit $\mathcal{C}$ of size $s$. In the black-box model, access to $\mathcal{C}$ is restricted to evaluation queries only (see \Cref{fig:black_box_poly}). Given such a circuit $\mathcal{C}$ and the degree bound $d = \deg(f)$, the goal is to compute the irreducible factors of $f$.

\begin{figure}
    \centering
    \begin{tikzpicture}[scale=0.6, every node/.style={font=\small}]
        \coordinate (A) at (-2,-2); 
        \coordinate (B) at (-2,2);  
        \coordinate (C) at (1,0);   

        \draw[fill=PeachPuff1, thick] (A) -- (B) -- (C) -- cycle;

        \node[font=\Large] at (-1,0) {$\calC$};

        \node[left=1.2cm of A] (alpha1) {$\alpha_1$};
        \node[left=1.2cm of B] (alphan) {$\alpha_n$};
        \draw[->, thick] (alpha1.east) -- ++(0.4,0) -- (-2.2,-1.2);
        \draw[->, thick] (alphan.east) -- ++(0.4,0) -- (-2.2,1.2);
        \node at (-2.5,0) {$\vdots$};

        \node[right=1.2cm of C] (output) {$f(\vec{\alpha}) \in \mathbb{F}$};
        \draw[->, thick] (C) -- (output.west);
    \end{tikzpicture}
    \caption{Black-box access to polynomial via circuit $\mathcal{C}$}
    \label{fig:black_box_poly}
\end{figure}

From now on, we will work over fields of characteristic zero unless otherwise specified. The results also hold if the characteristic is sufficiently large compared to the degree of the polynomial under consideration.

A special case of multivariate factorisation arises when the polynomial is a perfect power, i.e., $f = g^e$ where $g$ is an irreducible polynomial and $e \geq 1$.
For fields of large characteristic, the generalised binomial theorem gives a simple way to compute $g$ from the circuit computing $f$.

\begin{theorem}[Perfect Power Factorisation]
\label{thm:vp-closure-special}
    Let $\F$ be a field of characteristic zero and $f \in \F[\vecx]$ be an $n$-variate polynomial of degree $d$ computable by a circuit of size $s$.
    If $f = g^e$ for some irreducible polynomial $g$ and $e \geq 1$, then $g$ can be computed by a circuit of size $\poly(s, n, d)$.
\end{theorem}

\begin{proof-sketch}
    We begin by performing a linear transformation on $f$ to ensure that $f(0, \ldots, 0) = 1$.
    Then using the generalised binomial theorem, we have
    \begin{align}
        g \;=\; f^{1/e} \; &=\; \bigg(1 + (f - 1)\bigg)^{1/e} \nonumber \\
        & =\; \sum_{i = 0}^{d} \; \binom{1/e}{i} \cdot (f - 1)^i \bmod{\inangle{\vecx}^{d+1}}. \label{eq:vp-closure-special}
    \end{align}
    The last equality holds because $(f-1)^i$ contributes terms of degree strictly larger than $d$ for all $i > d$.
    Algebraic circuits can be efficiently added and multiplied, and further, the division required to compute the binomial coefficients in \Cref{eq:vp-closure-special} can be eliminated using \Cref{lem:division-elimination}.
    Overall, given only black-box access to $\calC$ computing $f$, we can obtain the circuit for $g$ of size $\poly(s, n, d)$.
\end{proof-sketch}

It is worth highlighting that the multiplicity $e$ can be exponential in the size $s$ of the circuit computing $f$ (e.g., constructed via repeated squaring).
Consequently, \Cref{thm:vp-closure-special} implies factor closure not just for special polynomials of $\VP$, but also for high-degree polynomials of the form $f = g^e$ computable by small circuits.
We will revisit this special case in \Cref{thm:vp-closure-special-NI}, where we discuss an alternative proof.
In the upcoming section, we will study the general factorisation problem, where the known results are not as strong as in this special case.

\paragraph{\bf Black-box multivariate factoring.}

Remarkably, Kaltofen and Trager~\cite{KT1990} showed that in the black-box setting, efficient circuit factorisation is achievable. To describe it, we need to state Hilbert's Irreducibility Theorem, which guarantees that irreducible factors of $f$ maintain their irreducibility profile when restricted to the random subspace. 
\begin{theorem}[Effective Hilbert Irreducibility Theorem]\label{thm:hilbert-irr}
    Let $S\subseteq \F$ be a large enough finite subset and $g \in \F[\vecx, y]$ be a monic polynomial of degree $d$ such that $\partial_y g$ is non-zero. If $|S|$ is at least $d^6$ and $g$ is irreducible, then for most of the $\overline{\alpha}, \overline{\beta} \in S^n$, the polynomial $g(\alpha_1 x + \beta_1, \dots, \alpha_n x + \beta_n, y)$ is irreducible.
\end{theorem}

\noindent The effective version of Hilbert's Irreducibility Theorem was first proved by Kaltofen~\cite[Section 3]{Kal1995}. For a clear and accessible exposition, see Sudan’s lecture notes~\cite[Lecture 9]{Sud1998}.

Consider the polynomial $f$ restricted to a suitable subspace so that the theorem above holds:
$$f_{\alpha, \beta} \;\coloneqq\; f\big(y, \alpha_1 x + \beta_1, \dots, \alpha_n x + \beta_n\big).$$ Let $\ell$ be the number of irreducible factors of $f_{\alpha, \beta}$ obtained using the bivariate factoring algorithm (see \Cref{sec:poly-factoring-dense}). Assuming $f$ is monic in $y$ and $\partial_y f$ is non-zero, one can prove using \Cref{thm:hilbert-irr} that $f$ has $\ell$ irreducible factors with high probability. Note that, since the bivariate factoring algorithm requires the input polynomial to be in dense representation, the coefficients of $f_{\alpha, \beta}$ are computed via interpolation (ref \Cref{lem:interpolation}).
 
Given a point $(a, b_1, \dots, b_n) \in \F^{n+1}$ and an index $i$, we want to compute the $i$-th factor $f_i(a, b_1, \dots, b_n)$. For this, define a trivariate polynomial which captures both the projection of $f$ on $\overline{\alpha}, \beta$ and the evaluation point $(a, b_1, \dots, b_n)$ as follows: 
$$\hat{f}(y, z_1, z_2) \;\coloneqq\; f\big(y, \alpha_1 z_1 + \beta_1 + (b_1 - \beta_1) z_2, \dots, \alpha_n z_1 + \beta_n + (b_n - \beta_n) z_2\big),$$ where each $x_i \mapsto \alpha_i z_1 + \beta_i + (b_i - \beta_i)z_2$. 
Note that $\hat{f}(a, 0, 1) = f(a, b_1, \dots, b_n)$ and $\hat{f}(y, x, 0) = f_{\alpha, \beta}(y,x)$. Using interpolation described in \Cref{lem:interpolation}, we can obtain the coefficients of $\hat{f}$. 
Kaltofen \cite{Kal1985} proved that factoring a constant-variate polynomial, in particular the trivariate polynomial $\hat{f}$, is efficiently doable using univariate and bivariate factorisations.
Let $\setdef{\hat{f}_i(y, x_1, x_2)}{ i \in [\ell']}$ be the set of irreducible factors of $\hat{f}$. Find an index $j$ such that $\hat{f}_j(y, x, 0) = \tilde{f_i}(y, x)$ and output $\hat{f}_j(a, 0, 1)$.

\paragraph{\bf Circuit factoring}

Black-box factoring is a strong notion of factorisation, wherein the algorithm is given oracle access to a polynomial and is required to construct oracle access to its factors---without relying on or even knowing the underlying representation or model of computation. Remarkably, this idea can be adapted to show that if a polynomial is computed by a small algebraic circuit, then all its irreducible factors can be computed by small circuits \cite{Kal1989}.
The complexity class $\VP$ contains all polynomial families that can be computed by `small' circuits. 

\begin{definition}[$\VP$]
    \label{def:vp}
    A polynomial family $f=(f_n)$ is in the class $\VP$ over the field $\F$ if both the number of variables and degree of $f_n$ are bounded by $\poly(n)$ and moreover, the size of the smallest circuit over $\F$ computing $f_n$, denoted $\size_{\F}(f_n)$ is bounded by $\poly(n)$. 
\end{definition}

Note that polynomials of degree $\exp(n)$ can be computed by circuits of size $\poly(n)$ by repeated squaring. Such families are not in $\VP$ due to the degree restriction in its definition. There are good reasons for imposing this restriction~\cite{Val1979,Gro2013}. As mentioned earlier, we will also discuss small circuits of high degree in \Cref{subsec:high-deg-ckts}. The main result of this section is the closure of $\VP$ under taking factors.

Over characteristic zero, we can without loss of generality assume that $f = g^e\cdot h$, where $g$ and $h$ are coprime and $e \geq 1$. The special case of $h=1$ can be handled as before using \Cref{eq:vp-closure-special}.

\begin{figure}
    \centering
    \begin{tikzpicture}[
    node distance=0.8cm and 1.2cm,
    every node/.style={font=\scriptsize},
    process/.style={
        rectangle, rounded corners,
        minimum width=2.2cm,
        minimum height=1.4cm,
        draw=black, fill=blue!10
    },
    startstop/.style={
        rectangle, rounded corners,
        minimum width=2.2cm,
        minimum height=1.4cm,
        draw=black, fill=gray!20
    },
    arrow/.style={->, thick},
    desc/.style={align=left, font=\tiny, text width=2.3cm}
    ]

    \node (input) [startstop] {Input: $f \in \F[\vecx,y]$};
    \node (preprocess) [process, right=of input] {Preprocessing};
    \node (univar) [process, below=of preprocess] {Univariate Factoring};
    \node (hensel) [process, left=of univar] {Hensel lifting};
    \node (linsys) [process, below=of hensel] {Solve Linear System};
    \node (gcd) [process, right=of linsys] {GCD Computation};
    \node (undo) [process, below=of gcd] {Undo Preprocessing};
    \node (output) [startstop, left=of undo] {Output: Factor $g$};

    \draw [arrow] (input) -- (preprocess);
    \draw [arrow] (preprocess) -- (univar);
    \draw [arrow] (univar) -- (hensel);
    \draw [arrow] (hensel) -- (linsys);
    \draw [arrow] (linsys) -- (gcd);
    \draw [arrow] (gcd) -- (undo);
    \draw [arrow] (undo) -- (output);

    \node [right=2.5mm of preprocess, desc] {Square-free reduction, monic transformation, and bivariate projection.};
    \node [right=2.5mm of univar, desc] {Factor $f(0,y)$ using univariate algorithms such as Berlekamp or Cantor--Zassenhaus.};
    \node [left=2.5mm of hensel, desc] {Apply Hensel lifting to iteratively lift factors modulo $y^{2^k}$.};
    \node [left=2.5mm of linsys, desc] {Solve a linear system to refine the lifted factor to match actual factor.};
    \node [right=2.5mm of gcd, desc] {Recover a non-trivial factor of $f$ using $\gcd$.};
    \node [right=2.5mm of undo, desc] {Undo preprocessing to return factor of original input polynomial.};

    \end{tikzpicture}
    \caption{Overview of Multivariate Factoring using Hensel lifting.}
    \label{fig:factoring-diagram}
\end{figure} 

\begin{theorem}[$\VP$ factor closure]\label{thm:vp-closure}
    Let $\F$ be a field of characteristic zero, and $f \in   \F[\vecx, y]$ be an $(n+1)$-variate, degree $d$ polynomial computable by a circuit of size $s$. If there are coprime factors $g$ and $h$ such that $f = g^e \cdot h$, where $e \geq 1$, then $g$ is also computable by a circuit of size $\poly(s, n, e)$.
\end{theorem}

\begin{proof-sketch}
    As in the univariate case, we first perform a square-free reduction to ensure $f$ has no repeated factors, by dividing $f$ by $\gcd(f, \partial_y f)$.
    The efficiency of this step follows from the structural results on algebraic circuits discussed in \Cref{subsec:struct-results}. 
    It is important to note that the multiplicity parameter $e$ appearing in the final size bound for the factor originates from this square-free reduction step.
    Next, as before we apply a linear transformation to make $f$ monic in $x$.
    To reduce the problem to bivariate factoring, we invoke the effective Hilbert Irreducibility Theorem (\Cref{thm:hilbert-irr}), effectively projecting the variables $\vecx, y$ onto a two-dimensional plane spanned by $x$ and $y$.
    This projection preserves the irreducibility profile of the factors.
    After this sequence of transformations, we work with the polynomial $f = g \cdot h \in \F[x, y]$, where $g$ and $h$ remain coprime and $f$ is monic in $y$.
    Furthermore, we can ensure that the restriction to $y=0$ preserves coprimality, meaning $g(x, 0)$ and $h(x, 0)$ are coprime.
    These conditions satisfy the prerequisites for Hensel lifting.
    Crucially, all these preprocessing steps are reversible and can be efficiently implemented using algebraic circuits.
    
    We iteratively apply Hensel lifting to compute polynomials $g_k$ and $h_k$ such that $f \equiv g_k \cdot h_k \pmod{y^{2^k}}$. The lifting proceeds until $2^k > 2d^2$ (refer to \Cref{prop:factor-existence}). Each iteration consists of basic algebraic operations outlined in \Cref{eq:monic-hensel-lift}, all of which can be performed efficiently by leveraging the structural results of algebraic circuits. 
    Since Hensel lifting yields only an \emph{approximate} factorisation modulo $y^{2^k}$, the actual factor $g$ of $f$ is recovered by solving the linear system $g' \equiv g_k \cdot \ell \pmod{y^{2^k}}$.
    
    By computing the $\gcd$ of $f$ with the solution $g'$ of the linear system, we recover a non-trivial factor $g$ of $f$. Finally, we reverse all preprocessing steps to obtain an algebraic circuit computing $g$.
\end{proof-sketch}

In \Cref{sec:lagrange-inversion}, we discuss a non-iterative proof of factor closure for $\VP$ based on a classical identity due to Lagrange. Extending this closure result to fields of small characteristic remains an open problem.

\begin{questype}{Open Problem}
\label{question:vp-factors-fq}
    Is the class $\VP$ closed under taking factors over fields of positive characteristic?
\end{questype}

As mentioned earlier, Kaltofen~\cite[Theorem 2]{Kal1987} proved a special case of factor closure for polynomials in $\VP$ of the form $g^e$ (i.e., $h = 1$ in \Cref{thm:vp-closure}) over fields of characteristic zero. Andrews~\cite{And2020} removed the dependence on characteristic in this result for the $\log$-variate regime, thereby making progress toward resolving \Cref{question:vp-factors-fq}. Later, we will discuss progress on this question in a more general setting, where factors of polynomials in $\VP$ over finite fields are shown to be \emph{explicit}.

There are other natural classes for which a factor closure result, similar to \Cref{thm:vp-closure}, does hold. Notably, these include classes that contain $\VP$, and also classes that are contained in $\VP$. We discuss these next.

\subsection{Algebraic branching programs} \label{subsec:abp}

Another polynomial computation model that is often studied is algebraic branching programs.

\begin{definition}[Algebraic Branching Program] \label{def:abp}
    An \emph{Algebraic Branching Program} (ABP) is a directed acyclic graph where the edges are labelled by linear polynomials over a field.
    The graph has a unique source node $s$ and a sink node $t$ (refer to \Cref{fig:abp}). 
    For every path $\pi$ from $s$ to $t$, let $f_{\pi}$ denote the product of the labels on the edges of $\pi$. The polynomial computed by the $\ABP$ is the sum of polynomials computed along all the different $s-t$ paths:

    $$f = \sum_{\pi: s \rightsquigarrow t} f_{\pi}.$$
    
    The \emph{size} of the ABP refers to the total number of vertices in the graph, while the \emph{length} of the ABP is the length of the longest path from $s$ to $t$.
\end{definition} 

\begin{figure}
    \centering
    \begin{tikzpicture}
    \node[circle,fill=LightSteelBlue1] (g14) at (10.0,-4) {$t$};
    
    \node[circle,fill=HotPink1] (g10) at (8.5,-2.5) {} edge[->] (g14);
    \node[circle,fill=HotPink1] (g11) at (8.5,-3.5) {} edge[->] (g14);
    \node[circle] (g12) at (8.5,-4.5) {$\vdots$};
    \node[circle,fill=HotPink1] (g13) at (8.5,-5.5) {} edge[->,thick] (g14);

    \node[circle] at (7.0,-2.5) {$\ldots$};
    \node[circle] at (7.0,-3.5) {$\ldots$};
    \node[circle] at (7.0,-5.5) {$\ldots$};

    \node[circle,fill=HotPink1] (g6) at (5.5,-2.5) {};
    \node[circle,fill=HotPink1] (g7) at (5.5,-3.5) {} edge[->,decorate,decoration=snake,thick] node[above=2mm] {$\pi$} (g13);
    \node[circle] (g8) at (5.5,-4.5) {$\vdots$};
    \node[circle,fill=HotPink1] (g9) at (5.5,-5.5) {};

    \node[circle,fill=HotPink1] (g2) at (4.0,-2.5) {} edge[->] (g6) edge[->,thick] (g7) edge[->] (g8);
    \node[circle,fill=HotPink1] (g3) at (4.0,-3.5) {} edge[->] (g6) edge[->] (g8);
    \node[circle] (g4) at (4.0,-4.5) {$\vdots$};
    \node[circle,fill=HotPink1] (g5) at (4.0,-5.5) {} edge[->] (g7) edge[->] (g9);
    
    \node[circle,fill=DarkSeaGreen1] (g1) at (2.5,-4) {$s$} edge[->,thick] (g2) edge[->] (g3) edge[->] (g5);

    \node[rectangle,scale=0.8] (f3) at (6.5, -2.0) {\emph{edge labels: linear polynomials.}};
    \node[rectangle,scale=0.8] (f2) at (6.5, -6) {$f_{\pi}=\text{prod.\ of edge labels in } \pi$};
    \node[rectangle,scale=0.8] (f1) at (6.5, -6.5) {$f =\sum_{\pi:s\rightsquigarrow t} f_{\pi}$};
    \end{tikzpicture}

    \caption{Algebraic Branching Program (ABP)}
    \label{fig:abp}
\end{figure}

Although not immediately evident, ABPs are equivalent to restricted algebraic circuits known as \emph{skew circuits}. In a skew circuit, each multiplication gate is allowed to have at most one child that is not an input gate. For a proof of equivalence between ABPs and skew circuits, see~\cite{Mah2014}. Similar to general circuits, ABPs are closed under addition and multiplication with only an additive increase in size. Strassen's classical technique for division elimination, originally developed for algebraic circuits, can be adapted to ABPs as well (see \Cref{lem:division-elimination}). 
We state a few non-trivial structural results about ABPs that will be instrumental in the upcoming factor closure result.

\begin{lemma}[ABP Closure Properties] \label{lem:abp-str}
    Let $f \in \F[\vecx, y]$ be an $(n+1)$-variate polynomial of degree $d$ computed by an ABP of size $s$, such that $f = \sum_{i=1}^d f_i(\vecx) y^i$.
    Then:
    \begin{enumerate}
        \item{\bf Coefficient Extraction.} For any $i \in [d]$, $f_i$ can be computed by an ABP of size $\poly(s, n, d)$.
        \item{\bf GCD.} For any polynomial $g$ computable by an ABP of size $s$, $\gcd(f, g)$ can be computed by an ABP of size $\poly(s, n, d)$. 
        \item{\bf Composition.} Let $g_1, \dots, g_n$ be polynomials computable by ABPs of size $s$. Then the composition $f(g_1, \ldots, g_n)$ can be computed by an ABP of size $\poly(s,n)$.
    \end{enumerate}
\end{lemma}

Refer to \cite[Section 2]{ST2021} for detailed discussion on these closure properties. Similar to the class of small degree polynomials computable by small circuits \Cref{def:vp}, we can define the class of polynomials computable by small ABPs.

\begin{definition}[$\VBP$] \label{def:vbp}
    A polynomial family $f=(f_n)$ is said to be in the class $\VBP$ if the number of variables in $f_n$ is bounded by $\poly(n)$ and furthermore, $f_n$ can be computed by an algebraic branching program of size at most $\poly(n)$.
\end{definition}

A polynomial computed by a polynomial-size algebraic branching program necessarily has a polynomially bounded degree.
Therefore, it is not hard to show that $\VBP \subseteq \VP$. 
In the previous section, we discussed the closure of the class $\VP$ under factorisation. A natural question is whether a similar closure property holds for $\VBP$. Kaltofen and Koiran~\cite{KK2008} made progress in this direction by showing that if $f, g \in \VBP$ and $f = g \cdot h$, then $h$ also belongs to $\VBP$. Their approach relies on interpreting algebraic branching programs (ABPs) as a restricted class of circuits known as skew circuits discussed earlier. 
Using this division closure, Jansen~\cite{Jan2011} leveraged the connection between ABPs and the determinant to adapt the power iteration algorithm for computing eigenvalues, to prove factor closure for polynomials of the form $f = f_1 \cdot (y - f_2) \cdots (y - f_n)$, where the $f_i$ are pairwise distinct.
More recently, Sinhababu and Thierauf~\cite{ST2021} revisited the Hensel lifting method to completely solve the open problem by proving factor closure of $\VBP$.

It is important to observe that the multivariate factoring algorithm from \Cref{thm:vp-closure} does not directly extend to Algebraic Branching Programs (ABPs). The main challenge lies in the monic variant of Hensel lifting, which requires division with remainder at each step---a costly operation for ABPs when repeated frequently (see \Cref{eq:monic-hensel-lift}). Although composition is technically feasible in ABPs (\Cref{lem:abp-str}), the multiplicative blow-up renders repeated use impractical. This limitation also rules out other factoring strategies that rely on efficient composition, which we shall explore in later sections.

Nevertheless, Sinhababu and Thierauf \cite[Theorem 4.1]{ST2021} circumvented the obstacles by employing the classical version of Hensel lifting (\Cref{thm:hensel}), which avoids division altogether.

\begin{theorem}[VBP Closure] \label{thm:vbp-closure}
    Let $f \in \F[\vecx]$ be an $n$-variate polynomial of degree $d$ computed by an ABP of size $s$. Then, every factor $g$ of $f$ can be computed by an ABP of size $\poly(s, n, d)$.
\end{theorem}
\begin{proof-sketch}
    The special case $f = g^e$ is handled via the generalised binomial expansion (analogous to \Cref{eq:vp-closure-special}), while leveraging the structural properties of ABPs discussed in \Cref{lem:abp-str}.
    For the general case $f=g^e \cdot h$ (where $e \geq 1$), we follow the same sequence of preprocessing steps as before: computing a derivative for square-free reduction, applying a random shift to ensure that $f$ is monic and that $g$ and $h$ remain coprime, and finally reducing the polynomial to a bivariate form.
    These steps prepare the polynomial for the application of Hensel lifting.
    As before, we begin the lifting process by applying a univariate factoring algorithm to $f(x, 0)$ (see \Cref{sec:univariate-factor}).
    However, instead of using the monic variant of Hensel lifting---as in the previous setting---which involves division at each step, we employ the classical version of Hensel lifting (see \Cref{thm:hensel}). 
    The process is repeated as before for $t$ iterations to obtain polynomials $g_t$ and $h_t$ such that $$f \equiv g_t \cdot h_t \pmod{y^{2^t}},$$ where $t = O(\log d)$ suffices. 
    Finally all that remains is to solve the linear system arising as before from $$\tilde{g} \equiv g_t \cdot h'_t \pmod{y^{2^t}}.$$ 
    Solving this linear system can be performed efficiently by an ABP, as the determinant of a matrix with polynomial entries is known to be in $\VBP$ (see \cite{MV1997}; for an exposition, see \cite[Section 3.3.3]{Sap2017}).
    Since a guaranteed solution requires the linear system to be homogeneous, the polynomial $\tilde{g}$ is not monic.
    Sinhababu and Thierauf observed that the monic irreducible factor $g$ can be recovered by dividing $\tilde{g}$ by its leading coefficient (see~\cite[Lemma 4.19]{ST2021}).
    It is then easy to observe that all the steps, including the preprocessing steps, can be performed efficiently on ABPs using the structural results discussed earlier.
\end{proof-sketch}

As for $\VP$, it is an open question to extend the above result to fields of small characteristic.

\begin{questype}{Open Problem} 
\label{question:vbp-factors-fq}
    Is the class $\VBP$ closed under taking factors over fields of positive characteristic?
\end{questype}

\subsection{Explicit polynomials}
\label{subsec:vnp-factoring}

A compelling class of polynomials was defined by Valiant~\cite{Val1979} as a non-deterministic analogue of $\VP$. This class is denoted by $\VNP$ and is defined as follows.

\begin{definition}[$\VNP$] \label{def:vnp}
    A family of polynomials $f=(f_n)$ is said to be in $\VNP$ over the field $\F$ if there exist functions $k,\ell,m:\N \to \N$ all polynomially bounded, and a polynomial family $g=(g_n) \in \VP$ with $g_n \in \F[x_1,\ldots,x_{k(n)},y_1,\ldots,y_{m(n)}]$ such that for all $n$,
    $$f_n(x_1,\ldots,x_{k(n)})=\sum_{\vecw \in \{0,1\}^{m(n)}} g_{\ell(n)}(x_1,\ldots,x_{k(n)},w_1,\ldots,w_{m(n)}).$$
\end{definition}

The polynomial family $g$ is often referred to as the family of \emph{verifier} polynomials, and the variables $y_1,\ldots,y_{m(n)}$ are called \emph{witness} variables. This is analogous to the Boolean class $\NP$ which contains functions $f(x_1,\ldots,x_n)$ that can be written as a logical OR of a verifier function $g(x_1,\ldots,x_n,y_1,\ldots,y_m)$, $m=\poly(n)$, over all Boolean assignments to the witness variables $y_1,\ldots,y_m$. 
The OR in $\NP$ is replaced by a summation in the definition of $\VNP$.
This structural shift aligns $\VNP$ naturally with $\#\P$, the complexity class that counts the number of accepting paths of a nondeterministic Turing machine, rather than merely asking if one exists.
This analogy is made rigorous by Valiant's seminal results, which established that permanent polynomial is complete for $\VNP$ \cite{Val1979}, and is also $\#\P$-complete \cite{ValPerm1979} in the Boolean setting.

Clearly, $\VP \subseteq \VNP$. Motivated by Kaltofen’s results on the factors of $\VP$, Bürgisser~\cite[Conjecture~2.1]{Bur2000} conjectured that $\VNP$ is closed under factorisation as well. Chou, Kumar, and Solomon~\cite[Theorem 2.9]{CKS2019} showed that \Burgisser's conjecture is in fact true when the field has characteristic $0$.
The key ideas were based on a related but different factoring tool that we are going to discuss in the upcoming section.
 
Further, Bhargav, Dwivedi, and Saxena~\cite[Theorem~1.6]{BDS2024} proved that the factor closure holds over finite fields as well, using Hensel lifting in combination with Valiant's criterion (and its converse) for \emph{low-degree} polynomials to be in $\VNP$ (cf.~\cite[Proposition~2.20]{Bur2000}). We denote by $\inangle{\gamma}$ the Boolean encoding of any mathematical object $\gamma$. 

\begin{proposition}[Valiant's Criterion] \label{prop:valiant-criterion}
    Consider a polynomial family $f=(f_n)$ with the number of variables and degree of $f_n$ bounded by $\poly(n)$. Suppose $f_n = \sum_{\vece} c_{\vece} \vecx^\vece$. If for every $n$, there exists a function $\phi_n \in \sPbyPoly$ such that given any exponent vector $\vece$, we have $\phi_n(\inangle{\vece}) = \inangle{c_{\vece}}$, then $f \in \VNP$.
\end{proposition}

Over finite fields, a weaker variant of Valiant's criterion suffices, wherein the coefficient function $\phi$ lies in $\#_p \mathsf{P} / \mathsf{poly}$ (see~\cite[Section~4.3]{Bur2000}). We omit this subtlety when it is clear from the context. 

Finally, it is once again important to state that $\VNP$ is closed under several standard operations, such as addition, multiplication, and composition~\cite[Section 4]{Val1982}. However, the closure under composition is particularly subtle. A naive substitution of variables may not work, and a more refined approach is needed~\cite[Claim 8.4]{CKS2019}. In particular, an alternative characterisation of $\VNP$ proves useful---namely, that $\VNP$ polynomials can be expressed as hypercube sums of polynomial-size formulas (see \cite[Theorem 2]{MP2008} for a proof). We refer the reader to the full version of~\cite{BDS2024} for detailed proofs.

\begin{theorem}[Closure of $\VNP$] \label{thm:vnp-closure}
    Let $f(\vecx)$ be an $n$-variate polynomial of degree $d$ such that
    $$f(\vecx) = \sum_{\vecw \in \{0,1\}^m} q(\vecx,\vecw),$$
    where $q(\vecx,\vecy)$ can be computed by a circuit of size $s$. If $g(\vecx)$ is any factor of $f$ of degree $r \leq d$, then it can be written as
    $$g(\vecx) = \sum_{\vece \in \{0,1\}^{m'}} u(\vecx,\vece),$$
    where $m'$ and $\size(u)$ are both bounded by $\poly(n,s,r,d,m)$.
\end{theorem}
\begin{proof-sketch}
    We sketch a unified proof covering both characteristic zero and finite fields of prime characteristic $p$.
    In doing so, we highlight the novel ideas from \cite{CKS2019}, which proved the result for characteristic zero, and \cite{BDS2024}, which extended it to the finite field setting.
    For the small characteristic case assume $\F = \F_q$, where $q = p^a$.

    We aim to show that if $f \in \VNP$ and $f = g^e \cdot h$ with $\gcd(g, h)=1$, then $g \in \VNP$. 
    Recall that in earlier proofs, the special case $h = 1$ was handled separately.
    For large fields, an identity of the form \Cref{eq:vp-closure-special} works in $\VNP$ as well.
    However, over finite fields, extra care is needed.

    Let $h=1$ and $e = p^i \cdot \hat{e}$, where $\gcd(\hat{e}, p) = 1$. We first extract the $\hat{e}$-th root. Define the following auxiliary polynomial with $g_1 \coloneqq g^{p^i}$:
    \begin{align*}
        \hat{F}(z) &\;\coloneqq\; z^{\hat{e}} - f \;=\; z^{\hat{e}} - (g^{p^i})^{\hat{e}} \\
        &= (z - g_1) \cdot (z^{\hat{e} - 1} + z^{\hat{e} - 2}g_1 + \cdots + g_1^{\hat{e} - 1}) \\
        &= (z-g_1) \cdot Q.
    \end{align*}
    Since $p \nmid \hat{e}$, it is easy to prove that $z - g_1$ is coprime to $Q$.
    We can therefore apply Hensel lifting to $\hat{F}(z)$.
    As observed in \cite[Lemma~5.9]{BDS2024}, the lifting process yields a small circuit $B$  and computes the factor $(z - g_1)$.
    Explicitly, this can be written as a composition:
    \[
        (z - g_1) \;=\; B\big(\hat{f}_0(\vecx), \dots, \hat{f}_d(\vecx)\big),
    \]
    where the inputs $\hat{f}_i(\vecx)$ are the coefficients of the polynomial obtained after applying the standard factoring preprocessing steps to $\hat{F}$.
    Crucially, each of these polynomials $\hat{f}_i(\vecx)$ remains in $\VNP$.
    Finally, using the closure of $\VNP$ under composition with polynomial-size circuits, we conclude that the factor $(z - g_1)$, and consequently $g_1$, belongs to $\VNP$.



    Therefore, the only remaining case is when $\hat{e} = 1$ and $f = g^{p^i}$.
    Note that this case requires separate attention only over small characteristic fields.
    Since the characteristic of the field is $p$, we can associate the coefficients of $f$ to those of $g$ via the standard Frobenius map.
    In particular, if the coefficient of $\vecx^{\veck}$ in $g$ is $c_g$, and the coefficient of $\vecx^{p^i \cdot \veck}$ in $f$ is $c_f$, then $c_g = (c_f)^{1/p^i}$.

    In~\cite{BDS2024}, the authors observed that the \emph{converse} of Valiant's criterion holds over finite fields.
    In particular, all the coefficients of $f$ can be computed by a $\sPbyPoly$ function.
    Together with the Frobenius correspondence between the coefficients, this yields a $\sPbyPoly$ function for computing the coefficients of the factor $g$.
    Then, using \Cref{prop:valiant-criterion}, we can conclude that $g$ can be written as a hypercube sum of small-sized circuits.
\end{proof-sketch}

\begin{remark}
    \Cref{thm:vnp-closure} also partially answers \Cref{question:vp-factors-fq} and \Cref{question:vbp-factors-fq}, proving that factors of polynomials in $\VP$ (and $\VBP$) are in $\VNP$ over finite fields. Robert Andrews communicated to us an (as yet) unpublished proof (based on the ideas of Malod and Portier~\cite{MP2008} and Andrews~\cite{And2020}) that extends the closure of $\VNP$ to all \emph{perfect fields}. A field $\F$ is perfect if either it is of characteristic zero, or, if the characteristic is a prime $p$, then all elements of $\F$ are $p$-th powers. All finite fields are perfect. Thus, it only remains to settle the closure of $\VNP$ over infinite non-perfect fields (e.g. the fraction field $\F_p(t)$ for some indeterminate $t$).
\end{remark}

\begin{questype}{Open Problem}
    Is the class $\VNP$ closed under taking factors over infinite non-perfect fields of positive characteristic?
\end{questype}


\section{Factoring via Newton iteration}
\label{sec:newton-iteration}

In 1669, Isaac Newton described a method to approximate roots of \emph{real-valued} functions. Also known as the Newton-Raphson method, the idea is to start with an initial approximate root $x_0$ of the equation $f(x)=0$ and successively produce better approximations using the rule 
$$x_{i+1}\; =\; x_i - \frac{f(x_i)}{f'(x_i)}\,,$$
provided $f'(x_i) \neq 0$. The value $x_{i+1}$ is the $x$-intercept of the tangent $y = f(x_i) + f'(x_i)(x-x_i)$ to the curve of $f(x)$ at the point $x_i$. The slope $f'(x_i)$ of the tangent keeps changing with the approximation (\Cref{fig:fast-ni}). Alternatively, this process can be viewed as finding a root of the Taylor approximation of $f(x)$ (up to linear terms) at the point $x_i$. It is intuitive that $x_{i+1}$ is closer than $x_i$ to a root of $f$.

\begin{figure}
    \centering
    \begin{tikzpicture}
    \def\func(#1){(#1)^3 - 1}
    \def\dfunc(#1){3*(#1)^2}
    \def\xZero{2.0}
    
    \pgfmathsetmacro{\yZero}{\func(\xZero)}
    \pgfmathsetmacro{\dZero}{\dfunc(\xZero)}
    \pgfmathsetmacro{\xOne}{\xZero - \yZero/\dZero}
    
    \pgfmathsetmacro{\yOne}{\func(\xOne)}
    \pgfmathsetmacro{\dOne}{\dfunc(\xOne)} 
    \pgfmathsetmacro{\xTwo}{\xOne - \yOne/\dOne}

    \begin{axis}[
        width=8cm, height=7cm,
        axis lines=middle,
        ymin=-2, ymax=9, xmin=0.5, xmax=2.5,
        xtick=\empty, ytick=\empty,
        xlabel={$x$}, ylabel={$f(x)$}
    ]
        \addplot[thick, blue, domain=0.5:2.3, samples=100] {\func(x)};
        
        \draw[dashed] (axis cs:\xZero,0) -- (axis cs:\xZero,\yZero);
        \draw[red, thick] (axis cs:\xZero,\yZero) -- (axis cs:\xOne,0);
        \node[below] at (axis cs:\xZero,0) {$x_0$};
        \node[below] at (axis cs:\xOne,0) {$x_1$};
        
        \draw[dashed] (axis cs:\xOne,0) -- (axis cs:\xOne,\yOne);
        \draw[red, thick] (axis cs:\xOne,\yOne) -- (axis cs:\xTwo,0);
        \node[below] at (axis cs:\xTwo,0) {$x_2$};
        
        \node[circle, fill=black, inner sep=1pt] at (axis cs:1,0) {};
    \end{axis}
\end{tikzpicture}
    \caption{Fast Newton iteration (dynamic slope)}
    \label{fig:fast-ni}
\end{figure}

Assuming the derivative $f'(x_0)$ at the initial guess is non-zero, a simpler rule
$$x_{i+1}\; =\; x_i - \frac{f(x_i)}{f'(x_0)}\,$$
also works. Unlike the previous case, the slope $f'(x_0)$ of the line $y=f(x_i) + f'(x_0)(x-x_i)$ (which is no longer the tangent) remains unchanged, and convergence to the root is slower (\Cref{fig:slow-ni}).

\begin{figure}
    \centering
    \begin{tikzpicture}
    \def\func(#1){(#1)^3 - 1}
    \def\dfunc(#1){3*(#1)^2}
    \def\xZero{2.0}
    
    \pgfmathsetmacro{\yZero}{\func(\xZero)}
    \pgfmathsetmacro{\dZero}{\dfunc(\xZero)}
    \pgfmathsetmacro{\xOne}{\xZero - \yZero/\dZero}
    
    \pgfmathsetmacro{\yOne}{\func(\xOne)}
    \pgfmathsetmacro{\xTwo}{\xOne - \yOne/\dZero}

    \begin{axis}[
        width=8cm, height=7cm,
        axis lines=middle,
        ymin=-2, ymax=9, xmin=0.5, xmax=2.5,
        xtick=\empty, ytick=\empty,
        xlabel={$x$}, ylabel={$f(x)$}
    ]
        \addplot[thick, blue, domain=0.5:2.3, samples=100] {\func(x)};
        
        \draw[dashed] (axis cs:\xZero,0) -- (axis cs:\xZero,\yZero);
        \draw[red, thick] (axis cs:\xZero,\yZero) -- (axis cs:\xOne,0);
        \node[below] at (axis cs:\xZero,0) {$x_0$};
        \node[below] at (axis cs:\xOne,0) {$x_1$};
        
        \draw[dashed] (axis cs:\xOne,0) -- (axis cs:\xOne,\yOne);
        \draw[red, thick] (axis cs:\xOne,\yOne) -- (axis cs:\xTwo,0);
        \node[below] at (axis cs:\xTwo,0) {$x_2$};
        
        \node[circle, fill=black, inner sep=1pt] at (axis cs:1,0) {};
    \end{axis}
\end{tikzpicture}
    \caption{Slow Newton iteration (fixed slope)}
    \label{fig:slow-ni}
\end{figure}

For simplicity of exposition, we will work over an algebraically closed field $\F$ of characteristic zero in this section, and remark about other fields when appropriate. Consider a multivariate polynomial $f(\vecx,y)$ over $\F$ whose root $g(\vecx)$ with respect to $y$ (i.e., $f(\vecx,g(\vecx))=0$) we want to find. Newton's method can be adapted to this case, provided the derivative is invertible and there is a meaningful notion of convergence. We view $f$ as a univariate polynomial $f \in \F(\vecx)[y]$ over the field of functions $\F(\vecx)$, and work over the formal power series ring $\F\llbracket \vecx \rrbracket$. The ring of formal power series $\F\llbracket \vecx\rrbracket$ can be viewed as the \emph{completion} of the polynomial ring $\F[\vecx]$ with respect to the maximal ideal $\langle \vecx \rangle$, which endows it with the $\langle \vecx \rangle$-adic topology. Similar to the notion of $p$-adic numbers (\Cref{subsec:p-adics}), two multivariate polynomials (or power series) $p(\vecx)$ and $q(\vecx)$ are considered \emph{close} up to degree $t$ if they only differ in terms of degree greater than $t$, i.e., $p(\vecx) = q(\vecx) \bmod \inangle{\vecx}^{t+1}$. We first state the slower version of Newton's method that is useful in factoring applications.

\begin{lemma}[\protect{\cite[Lemma 5.1]{CKS2019}}]
\label{lem:slow-ni}
    Let $f(\vecx,y) \in \F(\vecx)[y]$ be a polynomial over $\F(\vecx)$ and $\mu \in \F$ be such that $f(\mathbf{0},\mu)=0$ but $\delta \coloneqq \partial_y f(\mathbf{0},\mu) \neq 0 $. 
    
    Then, there is a \emph{unique} power series $\varphi \in \F\llbracket\vecx\rrbracket$ with $\varphi(\mathbf{0})=\mu$ such that $f(\vecx,\varphi)=0$. Set $\varphi_0=\varphi(\mathbf{0})$, and for all $t \geq 0$, define
    $$\varphi_{t+1}\; \coloneqq \; \varphi_t - \frac{f(\vecx,\varphi_t)}{\delta}.$$ 
    
    The rate of convergence of the sequence of \emph{polynomials} $\varphi_t \in \F[\vecx]$ to the root $\varphi$ is \emph{linear}:
    $$\varphi = \varphi_t \bmod \inangle{\vecx}^{t+1} \text{ for all } t \geq 0.$$
\end{lemma} 

Suppose that the polynomial $f(\vecx,y) \in \F[\vecx,y]$ has a degree-$r$ polynomial $g(\vecx) \in \F[\vecx]$ as a root (w.r.t.\ $y$). After $r$ steps of Newton iteration, we get a polynomial $\varphi_r$ (of degree possibly greater than $r$) that agrees with $g$ on all terms up to degree $r$. We obtain $g$ by truncating $\varphi_r$ up to terms of degree at most $r$. Seen differently, Newton iteration lets us find a \emph{linear} factor $(y-g(\vecx))$ of $f(\vecx,y)$ by iteratively approximating the root $g(\vecx)$, provided the root is of multiplicity one, also called a \emph{simple} root. If the root is not simple, i.e., $y-g(\vecx)$ occurs with multiplicity $e>1$ in $f$, the derivative at the root $\partial_y f(\vecx,g(\vecx))=0$, so we consider instead the $(e-1)$-th partial derivative $\partial_{y^{e-1}}f(\vecx,y)$, which has $g(\vecx)$ as a simple root.

\begin{remark}\label{rem:multiplicity}
    If the field $\F$ has characteristic $p$, taking derivatives to make the root simple would kill the polynomial, unless the multiplicity $e$ is coprime to $p$. In case $e=p^\ell e'$ where $e'$ is coprime to $p$, we can replace $y$ with $z^{1/{p^\ell}}$, making $g(\vecx)^{p^\ell}$ a root (w.r.t.\ $z$) of multiplicity $e'$, that \emph{can} be made simple by taking derivatives. Hence, we can only approximate the $p^\ell$-th power of $g$, and not $g$ itself.
\end{remark}

A `random' invertible shift of the variables $\vecx \mapsto \vecx + \boldsymbol{\alpha} y + \boldsymbol{\beta}$, where $\boldsymbol{\alpha}$ and $\boldsymbol{\beta}$ are random elements from $\F^n$, ensures that $f$ is monic in $y$ and thus, so are its factors (see, e.g., \cite[Lemma 14]{DSS2022}). However, in general, the factors of $f(\vecx,y)$ may not be linear. Nevertheless, over the algebraically closed field $\overline{\F(\vecx)}$, $f$ uniquely factors as $\prod_i (y-\varphi_i(\vecx))^{e_i}$ where for all $i$, $e_i>0$ and $\varphi_i(\vecx) \in \overline{\F(\vecx)}$. In fact, it can be shown that after the variable shift, $f$ splits over the power series ring i.e., $\varphi_i(\vecx) \in \F\llbracket\vecx\rrbracket$ for all $i$ (see \cite[Theorem 17]{DSS2022} for a proof). Moreover, the constant terms of the roots $\mu_i \coloneqq \varphi_i(\mathbf{0})$ are all \emph{distinct} non-zero field elements. Using Newton iteration, we can approximate these power series roots.

\begin{remark}\label{rem:roots-in-ext}
    If the field $\F$ is not algebraically closed, we might only find all the roots $\mu_i$ in an extension $\K \supseteq \F$ of degree $k!$ (at the most), where $k$ is the degree of $f$ with respect to $y$. We will then have that the power series roots $\varphi_i(\vecx) \in \K\llbracket\vecx\rrbracket$.
\end{remark}

We now state a version of Newton iteration with \emph{quadratic} rate of convergence. 

\begin{lemma}[\protect{\cite[Lemma 15]{DSS2022}}]
    \label{lem:fast-ni}
    Let $f(\vecx,y) \in \F(\vecx)[y]$ be a polynomial over $\F(\vecx)$ and $\mu \in \F$ be such that $f(\mathbf{0},\mu)=0$ but $\delta \coloneqq \partial_y f(\mathbf{0},\mu) \neq 0$. 
    
    Then, there is a \emph{unique} power series $\varphi \in \F\llbracket\vecx\rrbracket$ with $\varphi(\mathbf{0})=\mu$ such that $f(\vecx,\varphi)=0$. Set $\varphi_0=\varphi(\mathbf{0})$, and for all $t \geq 0$ define
    \[
        \varphi_{t+1}\; \coloneqq\; \varphi_t - \frac{f(\vecx,\varphi_t)}{\partial_y f(\vecx,\varphi_t)}.
    \]
    
    Since $\varphi_t(\mathbf{0})=\mu$ for all $t \geq 0$ and since $\delta \neq 0$, the denominator above is invertible in the power series ring. The sequence of \emph{rational functions} $\varphi_t(\vecx) \in \F(\vecx)$ converges to $\varphi$ at a \emph{quadratic} rate:
    \[
        \varphi = \varphi_t \bmod \inangle{\vecx}^{2^t} \text{ for all } t \geq 0.
    \]

\end{lemma}

The above Newton-Iteration lemmas can also be viewed as power series versions of the \emph{Implicit Function Theorem} (see, e.g.~\cite[Chapter 1]{KP2013} and \cite[Theorem 9.2.1]{Art2022}). Returning to the problem of factoring $f(\vecx,y)$, suppose it has a non-trivial factor $g(\vecx,y)$ of degree $r$. As discussed earlier, we can assume $g$ is a factor of multiplicity one by considering appropriate partial derivatives. Since $f$ is monic in $y$, so is $g$, and splits as $g = \prod_{i=1}^r (y-\varphi_i)$ for some power series roots $\varphi_i(\vecx) \in \F\llbracket\vecx\rrbracket$. Crucially, these $\varphi_i$'s are a subset of the roots of $f$. Given this set (via the \emph{distinct} degree-$0$ terms of the $\varphi_i$'s), we can build the approximate power series roots up to degree $r$ using Newton iteration on $f$. Finally, in order to obtain $g$, it is enough to multiply the approximate linear factors corresponding to $g$, truncate the result up to terms of degree $r$, and invert the variable shift. 

\begin{remark}
    The closure of a particular algebraic model under factoring does not require an efficient algorithm to \emph{find} the roots (factors), but only the \emph{structural} ability to perform Newton iteration (and other pre- and post-processing steps) within the model. Whether we use the faster or slower version of Newton's method will also depend on the convenience afforded by the model. We will have a chance to briefly comment on algorithms for factoring once we discuss the factor closure of a model. 
\end{remark}

\subsection{Low depth circuits}
\label{subsec:low-depth}

In this section, we consider factors of polynomials computable by small-size circuits that are also of \emph{bounded-depth}. These circuits form a fascinating subclass of arithmetic circuits since general circuits can be ``efficiently'' depth-reduced, unlike Boolean circuits. An influential body of work~\cite{VSBR1983,AV2008,Koi2012,Tav2015,GKKS2016} shows that any circuit of size $s$ computing a polynomial of degree $d$ can be equivalently written as a \emph{homogeneous} depth-$4$ circuit of size $s^{O(\sqrt{d})}$, or even a depth-$3$ circuit of size $s^{O(\sqrt{d})}$ (over fields of characteristic zero), if one foregoes homogeneity.

Dvir, Shpilka, and Yehudayoff~\cite{DSY2009} studied factors of low-depth circuits in order to extend the Hardness-Randomness paradigm of Kabanets and Impagliazzo~\cite{KI2004} (see \Cref{sec:hardness-randomness}) to this model. Of particular interest to them were factors of the form $y-g(\vecx)$ of a bounded-depth circuit $f(\vecx,y)$. They showed that if the individual degree of $f$ with respect to $y$ is bounded, then the root $g(\vecx)$ has a small bounded-depth circuit.

\begin{theorem}[\protect{\cite[Theorem 4]{DSY2009}}]
\label{thm:root-ideg-depth}
    Let $f \in \F[\vecx,y]$ be an $(n+1)$-variate polynomial with ${\deg(f) \leq d}$ and $\deg_y(f) \leq k$, computed by a depth-$\Delta$ circuit of size $s$. Then, its factors of the form $y-g(\vecx)$ with $\deg(g)=r$ can be computed by circuits of size $\poly(s,r^k,d,n)$ and depth $\Delta+O(1)$.
\end{theorem}
\begin{proof-sketch}
    The obvious approach would be to use \Cref{lem:slow-ni} to successively approximate the root $g(\vecx)$ via Newton iteration. Let us assume the preconditions for using the lemma are met. The crucial observation~\cite[Lemma 3.1]{DSY2009} is that the $t$-th approximate $\varphi_t$ can be expressed as a ($k+1$)-variate, degree-$t$ polynomial $Q_t$ in the coefficients of $f$, i.e., $\varphi_t = Q_t(f_0,\ldots,f_k)$, where $f(\vecx,y)=\sum_{j=0}^k f_j(\vecx)y^j$, with $f_j(\vecx) \in \F[\vecx]$. After $r$ iterations, we get the polynomial $\varphi_r = Q_r(f_0,\ldots,f_k)$ of degree at most $dr$ such that $g = \varphi_r \bmod \inangle{\vecx}^{r+1}$. 
    
    The sparsity of $Q_r$ is at most $\binom{r+k}{k} = O(r^k)$. Composing the trivial depth-$2$ circuit for $Q_r$ (of size $O(r^k)$) with the depth-$\Delta$ circuits for the $f_j$'s (of size $O(sk)$) obtained via interpolation (\Cref{lem:interpolation}) on $f$, we have a circuit for $\varphi_r$ of depth $\Delta+2$ and size $\poly(r^k,s)$. To extract $g$, we again use interpolation by first mapping $x_i \to zx_i$ and using~\Cref{lem:interpolation} to obtain all the components of $\varphi_r$ of degree at most $r$ (w.r.t.\ $z$) causing a size blow-up of $\poly(r^k,s,d)$. The depth does not change.
    
    By repeatedly taking partial derivatives of $f$ using \Cref{lem:partial-derivative} (incurring a size blow-up of at most $\poly(s,k)$) and a translation of the coordinates (resulting in a depth increase by $1$), we can ensure that $g$ is a simple root of $f$ (and thus, \Cref{lem:slow-ni} is applicable). So we have a depth $\Delta+3$ circuit for $y-g(x)$ of size $\poly(r^k,s,d,n)$.
\end{proof-sketch}  

\begin{remark}\label{rem:root-algo}
     
    In light of~\Cref{rem:multiplicity}, if the field $\F$ is of characteristic $p$, then~\Cref{thm:root-ideg-depth} only applies to $p$-th powers of linear factors. For a factoring \emph{algorithm}, we work over the field of rational numbers $\Q$, or a finite field. We note that the preprocessing steps to ensure applicability of Newton iteration can be performed efficiently, with access to randomness. The initial value $\varphi(\mathbf{0})=\mu$ (which is guaranteed to be in the field if a root $g$ exists) is found by factoring $f(\mathbf{0},y)$ using a univariate factoring algorithm, and each distinct $\mu$ lifts to a different root of $f$. Therefore, the arguments also yield a \emph{randomised} polynomial time algorithm to \emph{find} the linear factors of $f$~\cite[Theorem 4]{DSY2009}.
\end{remark}

Rafael Oliveira~\cite{Oli2016} generalised this result to show that polynomials of bounded individual degrees and computable by small bounded-depth circuits are closed under factorisation. A little thought reveals that if a polynomial has its individual degree bounded, then so will its factors.

\begin{theorem}[\protect{\cite[Theorem 1.2]{Oli2016}}]
\label{thm:factors-ideg-depth}
    Let $f \in \F[\vecx,y]$ be an $(n+1)$-variate polynomial with $\deg(f) \leq d$, and individual degree bounded by $k$. Suppose $f$ can be computed by a depth-$\Delta$ circuit of size $s$. Then, any factor $g(\vecx,y)$ of degree $r$ that divides $f$ can be computed by a circuit of size $\poly(r^k,k,s,d,n)$ and depth $\Delta+O(1)$.
\end{theorem}
\begin{proof-sketch}
    The main idea is to use the argument at the beginning of~\Cref{sec:newton-iteration} and essentially lift the above result of Dvir, Shpilka and Yehudayoff to all factors. Let us begin by writing $f(\vecx,y) = \sum_{j=0}^k f_j(\vecx)y^j$ and $g(\vecx,y) = \sum_{j=0}^{\ell} g_j(\vecx)y^j$, where $\ell \leq k$. A shift of the variables $\vecx \mapsto \vecx + \boldsymbol{\alpha} y + \boldsymbol{\beta}$ to ensure that $f$ is monic in $y$, however, would make its individual degree $kn$. Therefore, in order to reduce to the monic case, we rewrite 
    \begin{align*}
        f(\vecx,y) \;&=\; f_k(\vecx)\left(y^k+\sum_{j=0}^{k-1} f_j(\vecx)/f_k(\vecx) y^j\right) \text{ and } \\
        g(\vecx,y) \;&=\; g_{\ell}(\vecx)\left(y^{\ell}+\sum_{j=0}^{\ell-1} g_j(\vecx)/g_{\ell}(\vecx) y^j\right).
    \end{align*}
    
    The idea is that, since $g$ divides $f$, the leading term $g_{\ell}$ divides $f_k$. If we are able to recover the factor $g_\ell$ from $f_k$ by an induction on the number of variables (note the one fewer variable), and we somehow obtain a circuit for $g/g_\ell$ using Newton iteration for monic polynomial factoring, then multiplying the circuits for $g_\ell$ and $g/g_\ell$ would give us $g$.

    Although we can extract $f_k$ from $f$ via interpolation, we cannot afford the resulting size blow-up (by a factor of $k$) in the induction argument. To avoid this, we work with the polynomial $\tilde{f}(\vecx,y) \coloneqq y^k f(\vecx,1/y)$, which is the \emph{reversal} of (the coefficients of) the polynomial $f$. The leading coefficient of $\tilde{f}$ is now $f_0(\vecx) = f(\vecx,0)$ which has a smaller circuit than $f$. Moreover, the factors of the reversal of a polynomial are just the reversals of the original factors.  
    
    Once we find (by induction) a circuit of depth $\Delta+O(1)$ and size $\poly(r^k,s,d)$ for $g_0(\vecx)$ (the leading coefficient of $\tilde{g}(\vecx,y)$), we can use \Cref{thm:root-ideg-depth} to approximate the roots $\varphi_i(\vecx) \in \F\llbracket\vecx\rrbracket$ (up to degree $d$) of the monic rational function $\tilde{g}/g_0 = \prod_{i=1}^\ell (y-\varphi_i)$. Using Newton iteration on $\tilde{f}$, whose degree in $y$ is bounded by $k$, we can approximate these roots by circuits of depth $\Delta+O(1)$ and size $\poly(r^k,s,d)$. We further obtain a circuit of size $\poly(r^k,s,k,d)$ for $\tilde{g}/g_0$ by multiplying the approximations for the circuits and truncating using interpolation, resulting in a depth increase by at most a constant. Finally, we get $g(\vecx,y)$ by multiplying $g_0$ and $\tilde{g}/g_0$ and computing the reversal of $\tilde{g}$.
\end{proof-sketch}

\begin{remark}
    If the field $\F$ has characteristic $p$, then we only obtain factors up to a power of $p$ (see~\Cref{rem:multiplicity}). Over fields that are not algebraically closed,~\Cref{rem:roots-in-ext} implies that the circuit we obtain for the factors is over an extension field $\K \supseteq \F$ of degree at most $k!$. This circuit can be simulated over $\F$ with a $\poly(k!)$ size blow up, and depth increase by a factor of $O(\log k!)$~\cite[Theorem 4.2]{HY2011}, which is efficient when $k$ is bounded.

    The above result can be used to obtain a \emph{randomised} algorithm running in time $\poly(s,(nk)^k)$ to compute all the factors of the polynomial $f$~\cite[Remark 1.3]{Oli2016}. In addition to the ideas in \Cref{rem:root-algo}, we need to determine which power series roots correspond to a factor $g$ over $\F$, and this can be done by testing (deterministically) the $2^k$ possible subsets of the roots.
\end{remark}

Chou, Kumar, and Solomon~\cite{CKS2019} removed the restriction on individual degrees at the expense of requiring that the \emph{total degree} of the factor be small.

\begin{theorem}[\protect{\cite[Theorem 2.1]{CKS2019}}]
    \label{thm:cks-low-depth}
    Let $f \in \F[\vecx,y]$ be an $(n+1)$-variate polynomial of degree $d$. Suppose $f$ can be computed by a depth-$\Delta$ circuit of size $s$. Then, any factor $g(\vecx,y)$ of degree $r$ that divides $f$ can be computed by a circuit of size $\poly(r^{\sqrt{r}},n,d,s)$ and depth $\Delta+O(1)$.
\end{theorem}
\begin{proof-sketch}
    The broad outline is still the reduction of factoring to approximating roots, same as in the proof of~\Cref{thm:factors-ideg-depth}. Consider the simple case when $g$ is a root of $f$. If $\deg_{y}(f) \geq n$, the earlier observation by Dvir, Shpilka and Yehudayoff that there exists a $(\deg_{y}(f)+1)$-variate polynomial $\varphi_r$ that approximates $g$ up to degree $r$ becomes trivial since $g$ is already an $n$-variate polynomial. 
    
    However, irrespective of $\deg_y(f)$, Chou, Kumar, and Solomon show that $\varphi_r$ can in fact be written as an $(r+1)$-variate, $\poly(r)$-sized, degree-$r$ polynomial $A_r$ in the (at most) $r$-th order partial derivatives of $f$~\cite[Lemma 5.3]{CKS2019}\footnote{A proof of closure of $\VNP$ under factoring also follows from here using closure properties of $\VNP$ under coefficient extraction and composition (see \cite[Claim 8.3 and Claim 8.4]{CKS2019}).}. Using the depth-reduction results mentioned at the beginning of this section, we can further squash this to a depth-$3$ circuit for $A_r$ of size $r^{O(\sqrt{r})}$ when the field has characteristic zero. Composition with the depth-$\Delta$, $\poly(s,d,r)$-sized circuits for the partial derivatives results in a depth-$(\Delta+3)$ circuit for $\varphi_r$, from which we can extract $g$ using interpolation, as before.

    The rest of the argument to lift this to a general factor remains almost the same as in \Cref{thm:factors-ideg-depth}. The argument does not need special handling of the leading coefficient, as we do not need to worry about maintaining individual degree after the coordinate shift. 
\end{proof-sketch}

\begin{remark}
    If the degree $r$ of the factor is at most $O(\log^2 n/\log\log n)$, then the size bound for $g$ in \Cref{thm:cks-low-depth} above is $\poly(n)$.     
\end{remark}



Recently, there has been a flurry of activity on \emph{derandomising} the factorisation algorithms for constant depth circuits~\cite{KRS2024,KRSV2024,DST2024}. There has also been progress on computing GCD, resultants, and various other problems in linear algebra using circuits of constant depth~\cite{AW2024, BKRRSS-gcd2025}. Building on this, Bhattacharjee, Kumar, Ramanathan, Saptharishi, and Saraf~\cite{BKR+2025} gave the first subexponential time \emph{deterministic} algorithm for constant depth circuit factorisation. Soon after, the same authors along with Rai~\cite{BKRRSS2025} showed that if a polynomial can be computed by a depth $\Delta$ circuit of size $s$, then any of its factors can be computed by a circuit of size $\poly(s)$ and depth $\Delta+O(1)$, thereby positively answering the factorisation question for bounded depth circuits. We will sketch this breakthrough technique later in Section~\ref{subsec:formula-factoring}. We note that the positive characteristic case still remains open.

\begin{questype}{Open Problem}
    Let $f \in \F_q[\vecx]$ be a multivariate polynomial over a finite field $\F_q$ of prime characteristic $p$. If $f$ can be computed by a size $s$, depth $\Delta$ circuit, is it the case that any factor of $f$ can be computed by a circuit of depth $O(\Delta)$ and size $\poly(s)$?
\end{questype}
 
\subsection{High degree circuits}
\label{subsec:high-deg-ckts}

The families of polynomials we considered till now had `low' degree, bounded by a polynomial in the number of variables. Malod~\cite{Mal2003,Mal2007} considered the case when the degrees are not bounded.

\begin{definition}
    A family of polynomials $f=(f_n)$ is said to be in $\VP_{nb}$ over the field $\F$ if the number of variables and the size of the smallest circuit for $f_n$ (over $\F$) are bounded by $\poly(n)$.
\end{definition}

Note that the degree of a polynomial family in $\VP_{nb}$ can be \emph{exponential} in the size of the circuit. Kaltofen~\cite{Kal1987} showed (\Cref{thm:vp-closure}) that for an $n$-variate polynomial $f=g^eh$ where $g$ and $h$ are coprime and $g$ has multiplicity $e$, 
\begin{equation}
    \label{eq:factor-kaltofen}
    \size(g) = \poly(\size(f),\deg(g),e,n).
\end{equation}

Hence, the best size upper bound one can deduce for (exponential-degree) factors of polynomials in $\VP_{nb}$ from Kaltofen's result is exponential. This is unavoidable in general: the polynomial $x^{2^n}-1 = \prod_{i=1}^{2^n}(x-\xi^i)$ where $\xi$ is a $2^n$-th root of unity, can be computed by a circuit of size $O(n)$. Lipton and Stockmeyer~\cite{LS1978} showed that a random exponential-degree factor $\prod_{i \in S}(x-\omega^i)$ where $S \subset [2^n]$ and $|S|=\exp(n)$, requires $\exp(n)$-size circuits. So $\VP_{nb}$ is \emph{not} closed under taking factors\footnote{In fact, the same argument also shows the non-closure of $\VNP_{nb}$.}. Nevertheless, \Burgisser 's \emph{Factor Conjecture}~\cite[Conj. 8.3]{Bur2000} (stated below) is that the size of the factor $g$ in \Cref{eq:factor-kaltofen} should be independent of its multiplicity $e$. In particular, $\poly(n)$-degree factors of a polynomial family in $\VP_{nb}$ should be in $\VP$. 

\begin{questype}{Open Problem}
    \label{conj:factor-conj}
    Show that if $g(\vecx)$ is a factor of an $n$-variate polynomial $f(\vecx)$ then,
    $$\size(g) = \poly(\size(f),\deg(g),n).$$
\end{questype}

See \cite[Section 4]{Bur2004} for some applications of the above conjecture to decision complexity. Over the years, there has been partial progress towards a resolution. In the case when $f$ is a power of $g$, Kaltofen~\cite{Kal1987} (cf. \cite[Proposition 6.1]{Bur2004}) already showed the conjecture to be true (see~\Cref{thm:vp-closure-special} for an alternative proof).

\begin{theorem}[\protect{\cite[Theorem 2]{Kal1987}}]
\label{thm:vp-closure-special-NI}
    Suppose $f(\vecx)=g(\vecx)^e \in \F[\vecx]$ is an $n$-variate polynomial. Then,
    $$\size(g) = \poly(\size(f),\deg(g),n).$$ 
\end{theorem}
\begin{proof-sketch}
    We cannot directly approximate the roots of $f$ (and hence $g$) as outlined in \Cref{sec:newton-iteration} since we cannot reduce to the multiplicity one case by taking $(e-1)$ partial derivatives - the size blow up will be $\poly(\size(f),e)$. Instead, we find the root of the equation $y^e - f(\vecx)$ by Newton iteration (\Cref{lem:fast-ni}). 
    
    By shifting variables we can ensure that all the $O(\log(\deg(g)))$ steps can be performed. Crucially, we only need the homogeneous components of $f$ up to $\deg(g)$ for the intermediate steps, and any involved exponentiations can be performed using repeated squaring at most $O(\log e)$ times.  Straightforward analysis will show that
    $$\size(g) = \poly(\size(f),\deg(g),n,\log e),$$
    whereby noting that $e \leq \exp(\size(f))$ gives the result.
\end{proof-sketch}

As another special case of the Factor Conjecture, Dutta, Saxena, and Sinhababu~\cite{DSS2022} showed that the conjecture also holds if the \emph{square-free} part of $f$ is of low degree. Suppose $f=\prod_{i=1}^m f_i^{e_i}$ is the complete factorisation of $f$ with $f_i$'s irreducible. Recall that the square-free part of $f$ (also called the \emph{radical}) is the polynomial $\rad(f) \coloneqq \prod_{i=1}^m f_i$.

\begin{theorem}
    Let $f(\vecx,y)$ be an $n$-variate polynomial and let $g(\vecx,y)$ be a factor of $f$. We then have
    $$\size(g) = \poly(\size(f),\deg(\rad(f)),n).$$
\end{theorem}
\begin{proof-sketch}
    We follow the exposition of Sinhababu~\cite{Sin2019}. Consider the derivative
    $$
        \partial_y f(\vecx,y) = \partial_y \left(\prod_{i=1}^m f_i^{e_i}\right) = \sum_{i=1}^m e_i f_i^{e_i-1} \partial_y f_i \left(\prod_{j \neq i}f_j^{e_j}\right).
    $$
    Rewrite it as $\partial_y f = \prod_{i=1}^m f_i^{e_i-1} u$ with $u=\sum_{i=1}^m e_i \partial_y f_i \left(\prod_{j \neq i}f_j\right)$. The auxiliary polynomial $F \coloneqq f+z\partial_y f$ factors as $\prod_{i=1}^m f_i^{e_i-1} \left( \prod_{i=1}^m f_i + z \cdot u \right)$, where $z$ is a fresh variable. The factor $G \coloneqq \rad(f) + z\cdot u$ is coprime to $\prod_{i=1}^m f_i^{e_i-1}$, has the same degree as $\rad(f)$ and is of multiplicity one. Using Kaltofen's result (\Cref{thm:vp-closure}), we get $\size(G) = \poly(\size(f), \deg(\rad(f)),n)$. One can obtain $\rad(f)$ by setting $z=0$ in $G$, and any irreducible factor $f_i$ in size $\poly(\size(f), \deg(\rad(f)),n)$ by further factoring $\rad(f)$. The result follows by suitably combining powers of $f_i$ (via repeated squaring) to form $g$. 
\end{proof-sketch}

Shortly after proposing the factor conjecture, \Burgisser~\cite{Bur2004} showed its plausibility: low-degree factors can be \emph{approximated} by small circuits! We will discuss this next.

\subsection{Algebraic approximation}
\label{subsec:alg-approx}

A natural notion of approximation in Valiant's framework was studied by \Burgisser~\cite{Bur2004} (refer~\cite[Ch. 15, 20]{BCS1997} for historical motivation). A polynomial $f \in \F[\vecx]$ is said to be approximated by another polynomial $F \in \F[\varepsilon][\vecx]$ to an \emph{order of approximation} $M$ if 

\begin{equation}\label{eq:alg-approx}
    F(\vecx,\varepsilon) = \varepsilon^M f(\vecx) + \varepsilon^{M+1}\ Q(\vecx,\varepsilon),    
\end{equation}
for some polynomial $Q(\vecx,\varepsilon) \in \F[\vecx,\varepsilon]$. The \emph{approximate/border} size of $f$, denoted $\sizebar(f)$, is defined as the size of the smallest circuit \emph{over the ring of constants} $\F[\varepsilon]$ computing a polynomial $F$ that approximates $f$.

Over fields like $\R$ or $\C$, one can think of the above as an approximation in the sense $\lim_{\varepsilon \to 0} \varepsilon^{-M} F = f$. Another way of formulating the notion of approximation is to consider $F_{\varepsilon} \in \F(\varepsilon)[\vecx]$ over the rational function field $\F(\varepsilon)$ instead, and require an approximation of the form
$$F_{\varepsilon}(\vecx) = f(\vecx) + \varepsilon Q_{\varepsilon}(\vecx),$$ 
where $Q_{\varepsilon} \in \F[\varepsilon][\vecx]$ is a polynomial. The border complexity of $f$ is defined as before, but with the circuit size now calculated over $\F(\varepsilon)$. In this case, although $F_{\varepsilon=0}=f$ is defined at $\varepsilon=0$, the intermediate computations in the circuit for $F$ (over $\F(\varepsilon)$) might not be. Scaling arguments show that these two notions of approximation are equivalent~\cite[Lemma 5.6]{Bur2004}. For a detailed discussion on other natural definitions of approximation (both topological and algebraic), and their equivalence, we point the reader to~\cite[Section 5]{Bur2004}, \cite[Appendix 20.6]{BCS1997} and \cite[Section 2]{BIZ2018}. The notion of border complexity naturally suggests an approximation version of the class $\VP$.

\begin{definition}[$\VPbar$]
    A polynomial family $f=(f_n)$ is in the class $\VPbar$ if the number of variables of $f_n$, its degree, and the size of the smallest circuit over $\F[\varepsilon]$ \emph{approximating} $f_n$ are bounded by $\poly(n)$.
\end{definition}

It is clear that $\VP \subseteq \VPbar$. In an attempt to utilise sophisticated tools from algebraic geometry and representation theory to study Valiant's conjecture, Mulmuley and Sohoni~\cite{MS2001,MS2008} proposed $\VPbar \nsubseteq \VNP$ as the mathematically nicer (but possibly harder) conjecture. Further details on the \emph{Geometric Complexity Theory} (GCT) program for proving lower bounds can be found in~\cite{Reg2002,Mul2011,BLMW2011,Mul2012,Gro2012,Lan2017,BI2025}.

Returning to the factorisation problem,     \Burgisser\ showed that $\poly(n)$-degree factors of families in $\VP_{nb}$ are in $\VPbar$.

\begin{theorem}[\protect{\cite[Theorem 1.3]{Bur2004}}]
    \label{thm:factor-approx}
   Let $\F$ be a field of characteristic $0$ and let $f(\vecx,y)$ be an $(n+1)$-variate polynomial over $\F$. Suppose $g(\vecx,y)$ is a factor of $f$. Then,
   $$\sizebar(g) = \poly(\size(f),\deg(g),n).$$
\end{theorem}
\begin{proof-sketch}
    It suffices to show the result for irreducible $g$. Let $f=g^eh$ where $g$ and $h$ are coprime. As earlier, we would like to approximate the roots of $g$. Due to the possibly exponential multiplicity $e$, we cannot reduce to the case of a simple root by taking derivatives of $f$. We can, however, find a point $p$ (equal to $(\mathbf{0},0)$ after shifting coordinates), such that $g$ vanishes at $p$ but both $h$ and $\partial_y g$ do not~\cite[Lemma 3.3]{Bur2004}. We can then try to build the corresponding (unique) power series root $\varphi \in \F\llbracket\vecx\rrbracket$ of $g$ (which it shares with $f$) using Newton iteration (\Cref{lem:fast-ni}) on $f$. This seems quite straightforward, except that $\partial_y f = eg^{e-1}h\partial_y g + g^e \partial_y h$ always vanishes at $p$ when $e>1$. The main idea is to consider the \emph{perturbed} polynomial 
    $$F_{\varepsilon}(\vecx,y) \coloneqq f(\vecx,y+\epsilon) - f(\mathbf{0},\epsilon).$$ 
    Note that $F_{\varepsilon=0} = f$ and $F_{\varepsilon}(\mathbf{0},0)=0$. We encourage the reader to check that the derivative $\partial_y F_{\varepsilon}(\mathbf{0},0)$ does not vanish anymore, and hence we can use \Cref{lem:fast-ni} to construct the root $\psi_{\varepsilon}$ of $F_{\varepsilon}$, which in fact is an algebraic approximation of the root $\varphi$ of $f$ (and also $g$), i.e. $\psi_{\varepsilon=0} = \varphi$ (see \cite[Proposition 3.4]{Bur2004}). This way, we can efficiently \emph{approximate} the factor $g$.   
\end{proof-sketch}

Recently, an alternative proof of the above theorem was given in \cite{BKRRSS2025}, replacing the use of \Cref{lem:fast-ni} with a classical power series root identity, which provided a conceptually simpler analysis of the size of the approximating circuit. The approximate circuit constructed for the low-degree factor in \Cref{thm:factor-approx} has additional structure. Note that, a priori, the order of approximation $M$ for a polynomial $f$ can be arbitrarily large. However, \Burgisser~\cite[Theorem 5.7]{Bur2004,Bur2020} showed that over algebraically closed fields, $M=\exp(\sizebar(f))$. Therefore, the free constants from $\F[\epsilon]$ used for approximation can be assumed to be `only' exponential in degree and, hence, size. Since the circuit for the factor $g$ above is constructed via Newton iteration, the constants in $\F[\epsilon]$ are in fact circuits of size $\poly(\size(f))$ (but exponential in degree). Hence, imposing this restriction on the size of the univariate polynomials in $\epsilon$ is quite natural.

\begin{definition}[$\VPbar_\epsilon$]
    A polynomial family $f=(f_n)$ over the field $\F$ is in the class $\VPbar_{\epsilon}$ if the number of variables of $f_n$, its degree, and the size of the smallest circuit \emph{over} $\F$ approximating $f_n$ are bounded by $\poly(n)$.
\end{definition}

Note that the circuit size of the approximating polynomial is over the base field $\F$, thus also incorporating the size of the constants in $\F[\varepsilon]$. We call such circuits `presentable', and $\VPbar_\epsilon$ as \emph{presentable} $\VPbar$. Over finite fields, it can be shown that $\VPbar_{\epsilon}$ is contained in $\VNP$, a fact not known about $\VPbar$. 

\begin{theorem}[\protect{\cite[Theorem 1]{BDS2024}}]
    \label{thm:deborder-presentable}
    Let $\F$ be a finite field, and $f \in \F[\vecx]$ be an $n$-variate polynomial of degree $d$ such that the size of the smallest circuit (over $\F$) approximating $f$ is of size $s$. Then, $f$ can be written as 
    $$f(\vecx) = \sum_{\veca \in \{0,1\}^m} g(\vecx,\veca),$$
    where $m$, and $\size_{\F}(g(\vecx,\vecy))$ are bounded by $\poly(s,n,d)$.
\end{theorem}
\begin{proof-sketch}
    The main idea is to use Valiant's criterion (\Cref{prop:valiant-criterion}) to show that the polynomial $f = \sum_{\vece} c_{\vece} \vecx^\vece$ has coefficients $c_\vece$ that are ``easy to describe''. Notice that the polynomial $F$ approximating $f$ has a small circuit but exponential degree, where the high degree is only due to $\epsilon$.
    
    In the approximating expression for $f$ (\Cref{eq:alg-approx}), we interpolate $F$ on \emph{all} the variables (including $\epsilon$) using appropriate powers of unity. Consequently, each coefficient $c_\vece$ of $f$ can be written as a hypercube sum over a small-size high-degree circuit. When $\F$ is a finite field, we can \emph{simulate} this algebraic circuit using a small Boolean circuit. As a result, we further obtain that the coefficient function of $f$ is in $\#\P/\poly$, whence we can apply Valiant's criterion.
\end{proof-sketch}

As a corollary of \Cref{thm:deborder-presentable}, and the observation that the approximate circuit in \Cref{thm:factor-approx} is presentable, we can conclude that the families of low-degree factors of high-degree circuits are in $\VNP$, making partial progress towards the Factor Conjecture~(\Cref{conj:factor-conj}) which postulates that low-degree factors are in $\VP$.

\begin{corollary}
    Let $f \in \F_q[\vecx]$ be an $n$-variate polynomial over a finite field $\F_q$ of characteristic $p$. Suppose $g$ is a $\poly(n)$-degree irreducible factor of $f$ with multiplicity coprime to the characteristic $p$. If $\size(f) = s$, then $g$ can be written as
    $$g(\vecx) = \sum_{\veca \in \{0,1\}^m} u(\vecx,\veca),$$
    where $m$, and $\size(u(\vecx,\vecy))$ are bounded by $\poly(s,n,d)$.
\end{corollary}

The use of interpolation and Valiant's criterion as in~\Cref{thm:deborder-presentable} does not seem helpful over infinite fields, and it is an open problem to extend these results to all fields.

\begin{questype}{Open Problem}
    Can we show that $\VPbar_{\varepsilon} \subseteq \VNP$ over \emph{any} field $\F$?
\end{questype}



\section{Newton iteration to Hensel lifting and back}
\label{sec:hensel-newton}

Hensel lifting (\Cref{sec:hensel-lifting}) and Newton iteration (\Cref{sec:newton-iteration}), though different when viewed from afar, are quite similar upon closer inspection.

Let $f(\vecx,y)$ be a monic (with respect to $y$) polynomial of degree $d$ with a root $g(\vecx)$ (with respect to $y$) of multiplicity one, i.e. $f(\vecx,g(\vecx))=0$, but $\partial_y f(\vecx,g(\vecx)) \neq 0$. When we wanted to find iteratively better approximations to $g$, we started with an approximation $g_0$ that was consistent with $g$ modulo the ideal $\inangle{\vecx}$ and used Newton's method to obtain a better approximation $g_1$ that was consistent with $g$ modulo $\inangle{\vecx}^2$. The polynomial $g$ being a root of $f$ implies that $f$ can be factored as $f(\vecx,y) = (y-g(\vecx))h(\vecx,y)$ for some polynomial $h$. So, we can instead view our Newton iteration process as a Hensel lift of the factorisation $f = (y-g_0)h_0 \bmod \inangle{\vecx}$ to the factorisation $f = (y-g_1)h_1 \bmod \inangle{\vecx}^2$. Since $g$ was a root of multiplicity one, the factors $(y-g(\vecx))$ and $h(\vecx,y)$ were coprime. In this sense, Newton iteration is a special case of Hensel lifting.

In the other direction, we can also view Hensel lifting as a special case of (multivariate) Newton iteration. Suppose that we have the factorisation $f(\vecx,y) = g(\vecx,y) h(\vecx,y)$. Let us write $f = \sum_{i=0}^d f_i y^i$, and similarly $g = \sum_{i=0}^{d_1} g_i y^i$, $h = \sum_{i=0}^{d_2} h_i y^i$ where $f_i,g_i,h_i \in \F[\vecx]$ for all $i$, and $d_1 + d_2 = d$. Noting that $f_d=1$, we can view this factorisation as a system of equations:
\begin{align*}
    f_0 &= g_0 h_0 \\
    f_1 &= g_0 h_1 + g_1 h_0 \\
    \vdots \\
    f_{d-1} &= g_{d_1-1}h_{d_2} + g_{d_1}h_{d_2-1} = g_{d_1-1} + h_{d_2-1}\,.
\end{align*}

For a new set of $d$ variables $\{\vecu,\vecw\} = \{u_0,\ldots,u_{d_1-1}, w_0,\ldots,w_{d_2-1}\}$, consider the $d$ equations 
\begin{align*}
    \varphi_0 &\coloneqq f_0 - u_0 w_0 \\
    \varphi_1 &\coloneqq f_1 - u_0 w_1 - u_1 w_0 \\
    \vdots \\
    \varphi_{d-1} &\coloneqq f_{d-1} - u_{d_1-1} - w_{d_2-1}.
\end{align*} 

The coefficients $(g_0,\ldots,g_{d_1-1},h_0,\ldots,h_{d_2-1})$ are a common zero of the equations $\varphi = (\varphi_0,\ldots,\varphi_{d-1})$. Given an approximate root $(\veca,\vecb) \in \F[\vecx]^d$ modulo the ideal $\inangle{\vecx}$, a \emph{multivariate generalisation} of Newton iteration~\cite{vG2013} gives a better approximation 
$$(\veca^*,\vecb^*) = (\veca,\vecb) - J^{-1} \varphi(\veca,\vecb)$$
modulo the ideal $\inangle{\vecx}^2$, where the matrix of polynomials $J = (\partial_t \varphi_i)_{0 \leq i \leq d-1, t \in \{\vecu,\vecw\}}$ is the \emph{Jacobian} of $\varphi$.\footnote{If we want to approximate a root with respect to a single variable $y$ of a single equation $\varphi(\vecx, y)$, the Jacobian is a $(1 \times 1)$ matrix $[\partial_y \varphi]$, and this is just the usual Newton iteration.} We can now view a Hensel lift of the factorisation $f = g^{(0)} h^{(0)} \bmod \inangle{\vecx}$ to $f = g^{(1)} h^{(1)} \bmod \inangle{\vecx}^2$ as a \emph{multivariate} Newton iteration step of improving an approximate root to the above set of equations modulo $\inangle{\vecx}$, to modulo $\inangle{\vecx}^2$. The Jacobian and the Sylvester Matrix (\Cref{def:resultant}) are the same up to permutation of rows and columns~\cite[Lemma 3.2]{CKS2019a}. Hence, the invertibility of the Jacobian is equivalent to the coprimality of $g$ and $h$~\cite[Exercise 15.21]{vG2013}. Thus, Newton iteration generalises Hensel lifting. 

More generally, the two methods can be derived from one another over valuation rings. This is what von zur Gathen~\cite{von1984} has to say:

\say{\emph{Note that while Yun~\cite{Yun1976} motivates the Hensel method as a special form of the Newton method (``Hensel meets Newton"), here the Newton method is a corollary of the Hensel method (``Hensel beats Newton").}}

This folklore connection between Hensel lifting and Newton iteration has also appeared explicitly in a few places~\cite{Zip1981, von1984, Art2022}. Chou, Kumar, and Solomon~\cite{CKS2019a} used this connection to give a simplified proof of Kaltofen's $\VP$ closure result using the multivariate Newton iteration. For more restricted models, the method that is convenient while factoring depends on the efficiency of the factoring primitive that the model affords.  

\section{Factoring via Lagrange Inversion}
\label{sec:lagrange-inversion}

Given a power series $z=f(y)$ it is natural to ask if the inverse function $f^{-1}(z)$ has a power series as well. Lagrange was the first to attempt this {\em reversion of series} in the late 18th century to study planetary motion, which was generalised by Hans Heinrich B\"urmann.

\begin{theorem}[Lagrange inversion formula]
\label{thm:lagrange-inversion}
    Let $\F$ be a field of characteristic zero. Let $f(y) \in \F\llbracket y \rrbracket$ be a formal power series with a simple root at $y=0$, meaning $f(0)=0$ and $\partial_y f(0) \ne 0$.
    Then, the unique compositional inverse $g(z) \in \F\llbracket z \rrbracket$ satisfying $f(g(z)) = z$ is given by
    \begin{equation}\label{eq:inversion}
        g(z) = \sum_{n \ge 1} c_n z^n,
    \end{equation}
    where
    \[
        c_n = \frac{1}{n} \coef_{y^{n-1}} \left(\frac{y}{f(y)}\right)^n.
    \]
\end{theorem}

There are many proofs of the Lagrange inversion formula (see e.g.,~\cite[Section 5.4]{Sta2024}). We sketch the key insights of a straightforward derivation of the above theorem using complex analysis and contour integration~\cite{WW2021}. 

\begin{proof}[Proof Sketch by Complex Analysis]

Let $C$ be a small enough {\em circle} around $0$; let $u$ refer to a point in $C$. Without loss of generality,~$f$ is injective on $C$. Let $\zeta \coloneqq f(u)$ refer to a point in $f(C)$. It follows: $d\zeta = f'(u)\, du$ and $u = f^{-1}(\zeta) = g(\zeta)$.

$$g(z) \;=\; \frac{1}{2\pi i} \oint \limits_{f(C)} \frac{g(\zeta)}{\zeta-z} d\zeta \;=\; \frac{1}{2\pi i} \oint \limits_{C} \frac{uf'(u)}{f(u)-z} du \,.$$
We can write this as a formal power series in $z$,
$$g(z)\;=\; \sum_{n\ge0} z^n \left( \frac{1}{2\pi i} \oint \limits_{C} \frac{uf'(u)}{f(u)^{n+1}} du \right) \,.$$
Thus, $$c_n \,=\, \frac{1}{2\pi i} \oint \limits_{C} \frac{uf'(u)}{f(u)^{n+1}} du.$$
By {\em integration by parts}, it simplifies to
$$c_n\,=\, \frac{1}{2\pi i} \oint \limits_{C} \frac{1}{nf(u)^{n}} du.$$
By the {\em Residue Theorem}, we get
$$c_n\,=\, \frac{1}{n} \coef_{y^{-1}}\left( \frac{1}{f(y)^n}\right) \,=\, \frac{1}{n}\coef_{y^{n-1}}\left( \frac{y}{f(y)}\right)^n \,.$$
\end{proof}

The expression for the coefficient above is purely algebraic and holds over any field of characteristic zero, and not just the complex numbers. To that end, we now discuss a more elementary proof by Stanley \cite[First proof of Theorem 5.4.2]{Sta2024}.

\begin{proof}[Proof Sketch by Algebra]
    Let the compositional inverse be written as $g(z) = \sum_{m \ge 1} c_m \cdot z^m$. 
    Since $g(z)$ is the compositional inverse of $f(y)$, substituting $z = f(y)$ gives $g(f(y)) = y$. Substituting this into the series for $g$ yields:
    \[
        y \;=\; \sum_{m \ge 1} c_m \cdot f(y)^m
    \]
    Differentiating both sides with respect to $y$ gives:
    \[
        1 \;=\; \sum_{m \ge 1} m \cdot c_m \cdot f(y)^{m-1} \cdot \partial_y f(y)
    \]
    For any $n\in \N$, to isolate the $n$-th coefficient, we multiply both sides by $(y/f(y))^n$:
    \[
        \left(\frac{y}{f(y)}\right)^n \;=\; \sum_{m \ge 1} m \cdot c_m \cdot y^n \cdot f(y)^{m-n-1} \cdot \partial_y f(y).
    \]

    We can split the sum on the right-hand side into two parts: the terms where $m \neq n$, and the single term where $m = n$. Notice that when $m \neq n$, the expression $f(y)^{m-n-1} \partial_y f(y)$ is the derivative of $\frac{1}{m-n} f(y)^{m-n}$. 
    \[
        \left(\frac{y}{f(y)}\right)^n \;=\; \left( \sum_{m \neq n} \frac{m \cdot c_m}{m - n} \cdot y^n \cdot \partial_y \left( f(y)^{m-n} \right) \right) + n \cdot c_n \cdot y^n \cdot \frac{\partial_y f(y)}{f(y)}
    \]

    Next, we take the coefficient of $y^{n-1}$ on both sides. A fundamental property of formal Laurent series is that the derivative of any series never contains a $y^{-1}$ term. Therefore, the entire sum of derivatives vanishes to give:
    \[
        \coef_{y^{n-1}} \left(\frac{y}{f(y)}\right)^n \;=\; n \cdot c_n \cdot \coef_{y^{-1}} \frac{\partial_y f(y)}{f(y)}
    \]

    To evaluate the remaining residue on the right, recall that $f(y)$ has a simple root at $y=0$. We can write $f(y) = y \cdot h(y)$ where $h(0) \neq 0$. Using the product rule:
    \[
        \frac{\partial_y f(y)}{f(y)} \;=\; \frac{h(y) + y \partial_y h(y)}{y h(y)} \;=\; \frac{1}{y} + \frac{\partial_y h(y)}{h(y)}
    \]
    Since $h(0) \neq 0$, the term $\frac{\partial_y h(y)}{h(y)}$ is a standard power series with no negative powers. 
    Thus, the coefficient of $y^{-1}$ is exactly $1$. 
    Therefore,
    \[
        n \cdot c_n \; =\; \coef_{y^{n-1}} \left(\frac{y}{f(y)}\right)^n
    \]
\end{proof}

Lagrange Inversion can be used as the basis of a \emph{unifying} factoring technique for all the models of computation we have seen previously. As an important example, we consider algebraic \emph{formulas}, a model slightly weaker than circuits (\Cref{subsec:small-circuits}) and branching programs (\Cref{subsec:abp}), but nonetheless extremely natural.

\subsection{Small algebraic formulas \& Shallow circuits}
\label{subsec:formula-factoring}

The class $\VF$ is a natural restriction of $\VP$ where the graph underlying the circuit is a tree (\Cref{fig:alg-ckt-formula}). Since a formula of size $s$ can never compute a polynomial of degree greater than $O(s)$, we do not need to impose any additional degree restriction.

\begin{definition}[$\VF$]
    A family of polynomials $f=(f_n)$ is said to be in $\VF$ over the field $\F$ if the minimum size of the \emph{formula} (over $\F$) computing $f_n$ is bounded by a polynomial function in the number of variables $n$.
\end{definition}

\begin{remark}
    It is not hard to see that $\VF \subseteq \VBP \subseteq \VP$. None of the containments are known to be strict.
\end{remark}

The factors obtained via monic Hensel lifting (\Cref{subsec:small-circuits}) turn out to require polynomial divisions with remainder (essentially circuits), even if we start with a formula. We do not know how to convert circuits to formulas without a quasi-polynomial blow up in the formula size. The best result we knew in this direction, until recently, was the factor closure of quasi-polynomial sized formulas, first shown by Dutta, Saxena, and Sinhababu~\cite[Theorem 3]{DSS2022}. We state here the version from the work of Chou, Kumar, and Solomon~\cite{CKS2019} which gives an alternative proof of the same result using the ideas in \Cref{subsec:low-depth}. 

\begin{theorem}[\protect{\cite[Theorem 9.1]{CKS2019}}]
    Let $f(\vecx)$ be a polynomial of degree $d$ in $n$ variables which can be computed by a formula of size $s$. If $g(\vecx)$ is a degree-$r$ factor of $f$, then it can be computed by a formula of size $\poly(s,n,d,r^{O(\log r)})$.
\end{theorem}

To avoid divisions, a reasonable attempt would be to extend the techniques of Sinhababu and Thierauf for factoring branching programs (\Cref{subsec:abp}). All but the last step of solving a linear system (equivalently, computing a determinant) can indeed be done using small formulas. Unfortunately, we do not know of any $\poly(n)$-sized formulas for the $n \times n$ symbolic determinant. Thus the question of factor closure for formulas remained open, till a breakthrough appeared in 2025 \cite{BKRRSS2025}. This paper eliminated the need for iterations, unlike Newton or Hensel. It is based on a fundamental, analytic result on the power series root due to Furstenberg \cite[Prop.2]{Fur1967}. As the authors also note~\cite[Section 1.2.1]{BKRRSS2025}, we can alternatively use Lagrange inversion (\Cref{thm:lagrange-inversion}), whose surprising connection to root-finding we now present. 

\begin{theorem}[\protect{\cite[Corollary 1.6]{BKRRSS2025}}]
\label{thm-root-in-BKRRSS25}
    Let $\FF$ be a characteristic zero field and $F(x,y) \in \FF[x,y]$. We are interested in a {\em simple} root $\varphi(x)\in \FF\llbracket x \rrbracket$ such that ~$F(x,y) = (y-\varphi)\cdot Q(x,y)$ with $\partial_y F(0,0) =1$ and $\varphi(0)=0$. Then,
    \begin{equation}\label{eqn-Furstenberg}
        \varphi(x) \,=\, \sum_{m\ge1}\coef_{y^{m-1}}\left[\frac{(y-F(x,y))^m}{m}\right]\,.        
    \end{equation}
\end{theorem}

\begin{proof-sketch}
    The proof relies on identifying the root $\varphi(x)$ as a specific evaluation of a formal power series derived via the Lagrange inversion formula.
    Define the auxiliary function:
    \[
        z \;=\; \frac{y}{y - F(x,y)} \,=: f(y)\,.
    \]
    Rearranging this, we get $y = z \cdot (y - F(x,y))$. Observe that if we substitute $z=1$, we obtain $y = y - F(x,y)$, which implies $F(x,y) = 0$. Thus, evaluating the inverse function $y(z) \coloneqq f^{-1}(z)$ at $z=1$ should yield the well-defined root $\varphi(x) \in \F\llbracket x \rrbracket$.

    We apply \Cref{thm:lagrange-inversion} to express $y(z)$ as a power series in $z$.
    Note that $y/f(y) = y - F(x,y)$.
    Substituting this into the inversion formula (\Cref{eq:inversion}), we obtain:
    \begin{align}
        y(z) \;&=\; \sum_{m \ge 1} \frac{z^m}{m} \cdot \coef_{y^{m-1}}\left[ \left( \frac{y}{f(y)} \right)^m \right] \nonumber \\
        &=\; \sum_{m \ge 1} \frac{z^m}{m} \cdot \coef_{y^{m-1}}\left[ (y - F(x,y))^m \right] \in \F[x]\llbracket z \rrbracket. \label{eq:lif-application}
    \end{align}

    Now we try to `pass' from the power series convergence notion in $\FF[x]\llbracket z \rrbracket$ to that in $\FF\llbracket x \rrbracket$. To recover $\varphi(x)$, we must justify setting $z=1$ in \Cref{eq:lif-application}. This operation is valid in the ring $\FF\llbracket x \rrbracket$ if the series converges.
    
    Recall the hypothesis $\partial_y F(0,0) = 1$ and $F(0,0) = 0$.
    This implies that the linear term of $F(x,y)$ in $y$ is exactly $y$.
    Consequently, the difference $y - F(x,y)$ contains no constant term and no linear term in $y$ (when $x=0$).
    We can therefore write:
    \[
        y - F(x,y) \;\in\; \inangle{x, y^2}.
    \]
    Now, consider the term $(y - F(x,y))^m$. It is a sum of monomials of the form $x^a (y^2)^b$ where $a+b = m$.
    We are interested in the coefficient of $y^{m-1}$.
    For a term to contribute to $y^{m-1}$, we must have $2b \le m-1$.
    Therefore, the $x$-degree $a$ must satisfy $a \ge (m+1)/2$.
    
    Since the degree of $x$ grows with $m$, the series in \Cref{eq:lif-application} converges well in $\FF\llbracket x \rrbracket$ when $z=1$.
    Setting $z=1$ yields precisely the formula in \Cref{eqn-Furstenberg}, and by construction, this value satisfies $F(x, \varphi(x)) = 0$.
\end{proof-sketch}

With familiarity of the factoring techniques discussed in \Cref{sec:newton-iteration} and the structural results from \Cref{subsec:struct-results}, we can now prove \Cref{thm:vp-closure} in a single step rather than using iterative techniques. 
In addition to providing a simpler proof of known results, this \emph{magical} identity, in fact, resolves long-standing factor-closure questions. 
In what follows, we prove that bounded depth circuits are closed under factoring. We will first sketch the proof for algebraically closed fields of characteristic zero, where the argument is relatively straightforward, and then discuss how to extend it to arbitrary fields.

\begin{theorem}[\protect{\cite[Theorem 4.3]{BKRRSS2025}}]{\em\bf (Shallow circuit)}
\label{thm:const-dp-closure}
    Let $\F$ be an algebraically closed field of characteristic zero, and $f \in   \F[\vecx, y]$ be an $(n+1)$-variate, degree $d$ polynomial computable by a circuit of size $s$ and depth $\Delta$. 
    Then, any factor $g(\vecx,y)$ of $f$ is computable by a circuit of size $\poly(s, n, d)$ and depth $O(\Delta)$.
\end{theorem}

\begin{proof-sketch}
    The main structure of the proof follows the template from \Cref{sec:newton-iteration}. We begin by preprocessing $f$ so that all the roots $\varphi_i \in \F\llbracket\vecx\rrbracket$ in the factorisation $f(\vecx, y) = \prod_{i=1}^d (y-\varphi_i(\vecx))$ are simple. We can assume $g$ shares the first $r$ roots with $f$:
    \begin{equation}\label{eq:factor-g}
        g(\vecx,y) = \prod_{i=1}^r (y-\varphi_i(\vecx)). 
    \end{equation}
        
    In particular, $g(\mathbf{0}, y) = \prod_{i=1}^r (y-\varphi_i(\mathbf{0}))$ where $\mu_i \coloneqq \varphi_i(\mathbf{0}) \in \F$ since $\F$ is closed. In order to lift $\mu_i$ to $\varphi_i$, we would like to use \Cref{thm-root-in-BKRRSS25}. By scaling $f$, we can ensure that $\partial_y f(\mathbf{0},0)=1$. Notice that for any $\mu_i$, $\tilde{\varphi}_i \coloneqq \varphi_i - \mu_i$ is a simple root of $\tilde{f}_i(\vecx,y) \coloneqq f(\vecx, y+\mu_i)$ with constant term $\tilde{\varphi}_i(\mathbf{0})=0$. 
    
    We can now apply \Cref{thm-root-in-BKRRSS25} to get a power series expression for each of the $\tilde{\varphi}_i$'s and hence, also the $\varphi_i$'s. From the discussion in the latter part of the proof of \Cref{thm-root-in-BKRRSS25}, it can be seen that terms with index $m > 2r$ contribute only to degrees strictly greater than $r$ in the $\vecx$ variables. Hence, we can truncate $\varphi_i$'s to index at most $2r$. Given that $f$ (and thus, $\tilde{f}$) has a size $s$ circuit of depth $\Delta$, all these operations can be performed in size $\poly(s,n,d)$ and depth $\Delta+O(1)$ with the help of interpolation (\Cref{lem:interpolation}). Plugging in these circuits for the roots into \Cref{eq:factor-g} and extracting the components of degree at most $r$ gives us the result.
        
\end{proof-sketch}

As seen earlier (\cref{rem:roots-in-ext}), if the field $\F$ is not algebraically closed, the roots $\{\mu_i\}_{i \in [r]}$ of $g(\mathbf{0}, y)$ come from a very high degree splitting field $\K \supseteq \F$. Notice that the circuit for the power series roots $\{\varphi_i\}$ (and hence, also $g$) just uses constants from $\F$ if we treat the roots $\{\mu_i\}$ as input variables. Moreover, the coefficient of $y^j$ in \Cref{eq:factor-g} is the $j$-th elementary symmetric polynomial $e_j(\varphi_1,\ldots,\varphi_r) = \sum_{S \subseteq [r], |S|=j} \prod_{i \in S} \varphi_i$ in the power series roots. One can now use a result by Andrews and Wigderson~\cite[Theorem I.8]{AW2024} on efficiently computing elementary symmetric polynomials to obtain a circuit for these coefficients (and hence $g$) over $\F$. This is what the authors of \cite{BKRRSS2025} do. We describe a more direct approach based on a related recent result by the same authors~\cite{BKRRSS-gcd2025}.

We view $g(\vecx,y)$ (which is truly a polynomial over $\F$) as a polynomial (over $\F[\vecx,y]$) in the \emph{root variables} $\vecmu \coloneqq (\mu_1,\ldots,\mu_r)$. Note that $g$ is \emph{symmetric} with respect to the roots $\{\mu_i\}$: any permutation of the roots merely permutes the factors of the product in \Cref{eq:factor-g}. The fundamental theorem of symmetric polynomials (e.g., see~\cite[Theorem IV.6.1]{Lang2002} and~\cite{BC2017}) states that there is a polynomial $G$ (over $\F[\vecx,y]$) such that $g(\vecmu) = G(e_1(\vecmu), \ldots, e_r(\vecmu))$, where $e_j(\vecmu)$ is the $j$-th elementary symmetric polynomial. Crucially, the $e_j(\vecmu)$'s are the (signed) coefficients of $g(\mathbf{0},y)$. Therefore, we can instead write $g$ as a polynomial in these coefficients, which are in the base field $\F$. We know from the proof of \Cref{thm:const-dp-closure} that $g$ has a circuit of size $\poly(s, n, d)$ and depth $O(\Delta)$, and as it turns out, so does $G$. This is the bounded-depth analogue of a similar result by Bl\"{a}ser and Jindal~\cite[Theorem 4]{BJ2019} for general circuits. 


\begin{theorem}[\protect{\cite[Theorem 1.2]{BKRRSS-gcd2025}}]
\label{thm:sym-poly-complexity}
    Let $g \in \F[\vecx, y][\vecmu]$ be a polynomial of degree $d$ which is symmetric in $\vecmu$ variables and computable by a circuit of size $s$ and depth $\Delta$. Let $G \in \F[\vecx, y][\vecz]$ be the unique degree $d$ polynomial such that $f = G(e_1(\vecmu), \dots, e_r(\vecmu))$, where $e_1, \dots, e_r$ are the elementary symmetric polynomials. Then $G$ is computable by a circuit of size $\poly(s,r,d)$ and depth $\Delta + O(1)$.
\end{theorem}

\begin{remark}
    We note that \cite[Theorem 1.2]{BKRRSS-gcd2025} assumes the symmetric polynomial is an element of $\F[\vecmu]$. The slight generalisation presented here follows directly from their proof and is more convenient  for our application.
\end{remark}

The above theorem shows that $G$ also has a circuit of size $\poly(s, n, d)$ and depth $\Delta + O(1)$ over $\F[\vecx,y]$. All that remains is to undo the preprocessing steps to obtain the factors of $f$ over the field $\F$. 


\begin{remark}
    There is a version of Lagrange's inversion formula for fields of positive characteristic as well (see e.g., Theorem 1 and Remark 2 in~\cite{Hu2016}). If the field $\F$ is polynomially large in $snd$, \Cref{thm:sym-poly-complexity} also holds. But due to inseparability issues (\Cref{rem:multiplicity}), if the field $\F$ is of characteristic $p$, then the closure only holds up to $p$-th powers~\cite[Theorem 5.1]{BKRRSS-gcd2025}.
\end{remark}

Inspecting the proof of \Cref{thm:const-dp-closure}, one can see that it also works for more general models like formulas, branching programs, and circuits. In particular, we have the closure of $\VF$ under factorisation.

\begin{theorem}[\protect{\cite[Theorem 1.1]{BKRRSS2025}}]
\label{thm:formula-closure}
    Let $\F$ be a field of characteristic zero, and $f$ be an $n$-variate, degree $d$ polynomial over $\F$ computed by an algebraic formula of size $s$. Then, any factor $g$ of $\F$ can be computed by a formula of size $\poly(s,n,d)$ over $\F$.
\end{theorem}

\begin{questype}{Open Problem}
    Is the class $\VF$ closed under taking factors over fields of positive characteristic?
\end{questype}

\section{Factoring via Convex Geometry}
\label{sec:geometry}

We will now study a model that is \emph{not} closed under factoring in general, but has been a natural and elegant crucible of factoring ideas for decades. 

\subsection{Sparse polynomials}
\label{subsec:sparse-factorisation}

\emph{Sparse polynomials} are perhaps the simplest class that one can study. The sparsity of a polynomial $f$, denoted by $\norm{f}$, is the number of monomials in it. Alternatively, one can think of the sparsity as the size of a $\Sigma\Pi$ circuit computing $f$. The work of von zur Gathen and Kaltofen~\cite{vK1985}, which initiated the study of factoring sparse polynomials, gave a \emph{randomised} algorithm that outputs a factor in time polynomial in the sparsity of the factor. Naturally, it makes sense to ask if factors of sparse polynomials are sparse. Unfortunately, as they showed, this is \emph{not} true in general. The sparsity of a factor can be superpolynomial in the sparsity of the original polynomial.

\begin{example}[\protect{\cite{vK1985}}]
\label{eg:quasi-poly-blowup}
    The polynomial $f = \prod_{i=1}^n (x_i^d - 1)$ has sparsity $\norm{f} = 2^n$, but one of its factors $g = \prod_{i=1}^n (1+x_i+\ldots + x_i^{d-1})$ has sparsity $\norm{g} = d^n = \norm{f}^{\log d}$.
\end{example}

If the individual degree is comparable to the number of variables, then the blow-up can be exponential.

\begin{example}[\protect{\cite{BSV2020}}]
\label{eg:exp-blowup}
    Over the field $\F_p$, the polynomial $f = \sum_{i=1}^n x_i^p$ has sparsity $\norm{f} = n$, but the factor $g = \left( \sum_{i=1}^n x_i \right)^d$, for $0 < d < p$ has sparsity $\norm{g} = \binom{n+d-1}{d} \approx \norm{f}^d$.
\end{example}

The above examples are essentially the ``worst'' that we have, and it is natural to wonder whether they are indeed the worst \emph{possible}. One might then still hope that factors of a polynomial with \emph{bounded} individual degree are sparse. 

\begin{questype}{Open Problem}
\label{conj:sparse-fac}
    Let $f = g \cdot h$ be a polynomial with bounded (constant) individual degree. Then, does the following hold: $\norm{g} = \poly(\norm{f})$?
\end{questype}

Tools from convex geometry have been useful in recent progress in studying this question. For a polynomial $f = \sum_{\vece} c_{\vece} \vecx^\vece$, consider the set of exponent vectors in its support, 
$$\sup(f)\; \coloneqq\; \{\vece : c_{\vece} \neq 0\} \subseteq \Z^n.$$ 

The convex hull of the points in the support denoted $\Conv(\sup(f))$ is called the \emph{Newton Polytope} corresponding to $f$:

$$P_f \coloneqq \Conv(\sup(f)) = \left\{\sum_{\vece} \alpha_{\vece} \vece : \sum_{\vece} \alpha_{\vece} = 1, 0 \leq \alpha_{\vece} \in \R, \vece \in \sup(f)\right\} \subseteq \R^n.$$

A \emph{vertex} of $P_f$ is a point in the polytope that \emph{cannot} be written as a \emph{non-trivial} convex combination (i.e., $\alpha_{\vece} < 1$ for all $\vece$) of points in $P_f$. We will denote the vertex set of a polytope $P$ by $V(P)$. Ostrowski~\cite{Ost1999} observed that for polynomials $g$ and $h$,
$$P_{gh}\; =\; P_g + P_h,$$
where the addition is a \emph{Minkowski sum}, consisting of all points $\veca + \vecb$, such that $\veca \in P_g$ and $\vecb \in P_h$. It can be shown that
\begin{equation}\label{eq:newt-poly}
    \max\{{|V(P_g)|,|V(P_h)|}\}\; \leq\; |V(P_g + P_h)|\; \leq\; |V(P_g)| \cdot |V(P_h)|.    
\end{equation}

The upper bound on $|V(P_g + P_h)|$ is straightforward. For the lower bound, see~\cite[Proposition 3.2]{BSV2020} and Appendix K in the book of Schinzel~\cite{Sch2000}. This suggests a way to prove sparsity bounds. For a polynomial $f=gh$, if we can show $g$ is `dense' by showing a lower bound on $|V(P_g)|$, then by the above inequality, we also get a lower bound on $|V(P_g + P_h)|$, and in turn the sparsity of $f$. In other words, if $f$ is sparse, so is $g$.

Consider the case when $f=gh$ and $g$ is \emph{multilinear}, i.e., the individual degree of $g$ is at most $1$. A moment's thought shows that every monomial of $g$ corresponds to a vertex, i.e., $\sup(g)=V(P_g)$. 
Combining this with the lower bound in \Cref{eq:newt-poly}, we get
\begin{equation}
    \label{eq:mult-sparsity}
    \norm{f}\; \geq \;|V(P_f)| \;=\; |V(P_g + P_h)| \geq |V(P_g)|\; =\; \norm{g}.
\end{equation}

Therefore, if $f=gh$ is a multilinear polynomial, then $g$ is also multilinear, and we get $\norm{g} \leq \norm{f}$ from above, showing that sparse multilinear polynomials are closed under factoring. Shpilka and Volkovich~\cite{SV2010} gave an efficient \emph{deterministic} algorithm for factoring sparse multilinear polynomials. Volkovich~\cite{Vol2015} used the sparsity bound in \Cref{eq:mult-sparsity} to first extend their result to sparse polynomials that split into multilinear factors. In a later work~\cite{Vol2017}, he proved that factors of \emph{multiquadratic} polynomials are also sparse, and gave an efficient deterministic algorithm to factor such polynomials.

However, in general, the size of the vertex set $|V(P_g)|$ could be much smaller than the sparsity $\norm{g}$ of the polynomial. The polynomial $g=\left( \sum_{i=1}^n x_i \right)^d$ in \Cref{eg:exp-blowup} has only $n$ vertices in its Newton Polytope (corresponding to $x_1^d,\ldots, x_n^d$), but has $O(n^d)$ monomials. If the individual degree of a polynomial is bounded, Bhargava, Saraf, and Volkovich~\cite{BSV2020} showed that the vertex set is not too small either.

\begin{theorem}
\label{thm:vertex-bd}
    Let $\F$ be an arbitrary field, and let $g$ be a polynomial over $\F$ in $n$ variables of individual degree $d$. Then,
    $$|V(P_g)| \;\geq\; \norm{g}^{1/O(d^2\log n)}.$$
\end{theorem}

The bound is tight with regard to dependence on $n$ (see \cite[Claim 4.4]{BSV2020}). As an immediate corollary of the above theorem, we deduce a sparsity bound for the factors of polynomials with bounded individual degrees.

\begin{corollary}\label{cor:ideg-sparsity-bd}
    Let $f=gh$ be a polynomial in $n$ variables of individual degree $d$ and sparsity $\norm{f}=s$. Then, the sparsity of $g$ is $\norm{g} \leq s^{O(d^2\log n)}$. 
\end{corollary}
\begin{proof}
    Note that the individual degree of $g$ is bounded by $d$ as well. Using \Cref{eq:newt-poly} and \Cref{thm:vertex-bd}, we get
    $$\norm{f} \geq |V(P_f)| = |V(P_g + P_h)| \geq |V(P_g)| \geq \norm{g}^{1/O(d^2\log n)},$$
    and thus, the required bound.
\end{proof}

Although the sparsity bound of \Cref{cor:ideg-sparsity-bd} is worse than \Cref{eg:exp-blowup}, it can be viewed as affirmative evidence towards \Cref{conj:sparse-fac}. Using the above sparsity bound, one can also obtain a deterministic \emph{quasi-polynomial} time algorithm for sparse polynomials of bounded individual degree~\cite[Theorem 2]{BSV2020} (also see~\cite{HG2023}). We now give a brief overview of the ideas in the proof of \Cref{thm:vertex-bd}. 

\begin{proof-sketch}[\Cref{thm:vertex-bd}]
    The classical \Caratheodory's theorem of convex geometry states that any point $\vecx \in \Conv(X)$ in the convex hull of a set of points $X \subseteq \R^n$ can be written as a convex combination of at most $n+1$ points in $X$. In fact, the convex combination only needs to use the vertices of $\Conv(X)$. We need much fewer than $n+1$ points if we are okay with \emph{approximating} $\vecx$ by a convex combination. 
 
    In the approximate version of \Caratheodory's theorem~\cite[Theorem 3.6]{BSV2020}, we have a set $X$ with bounded $\ell_{\infty}$ norm (i.e., $\max_{\vecy \in X} \norm{\vecy}_{\infty} \leq 1$). Given an $\varepsilon > 0$, any point {$\vecx \in \Conv(X)$} that lies in the convex hull of $X$ can be \emph{$\varepsilon$-approximated} by a point $\vecx'$ (i.e., $\norm{\vecx-\vecx'}_{\infty} \leq \varepsilon$), and moreover, $\vecx'=\sum_\vecy \alpha'_\vecy \vecy$ is a \emph{uniform} convex combination ($\alpha'_\vecy = 1/k$ whenever $\alpha'_\vecy \neq 0$) of at most $k=O(\log n/\varepsilon^2)$ vertex points from $V(\Conv(X))$. The bound on $k$ follows from sampling points $\vecy \in V(\Conv(X))$ with probability $\alpha_\vecy$ where $\vecx = \sum_\vecy \alpha_\vecy \vecy$, and then using Chernoff's inequality to show that sampling $O(\log n/\varepsilon^2)$ points and taking their uniform convex combination $\varepsilon$-approximates $\vecx$. For us, $X$ will be the scaled-down version of $\sup(g)$, i.e., 
    $$X\; =\; \{1/d\cdot \vece : \vece \in \sup(g)\}.$$
    Note that $|X| = \norm{g}$ and $|V(\Conv(X))| = |V(P_g)|$. Moreover, for distinct $\vecx \neq \vecy \in X$, we have $\norm{\vecx-\vecy}_{\infty} \geq 1/d$. Hence, if we choose $\varepsilon$ to be something slightly smaller than $1/2d$, by the triangle inequality, a point $\vecx' \in \Conv(X)$ can $\varepsilon$-approximate only one of $\vecx$ or $\vecy$, not both. So, every point $\vecx \in X$ is guaranteed a distinct point $\vecx' \in \Conv(X)$ that approximates it \emph{and} is a uniform convex combination of at most $k = O(\log n/\varepsilon^2)$ points of $V(\Conv(X))$. The number of such $\vecx'$ is at most $|V(\Conv(X))|^k$. Hence, $|X| \leq |V(\Conv(X))|^k$, and we can choose $\varepsilon$ sufficiently close to (but smaller than) $1/2d$ to get $k=O(d^2 \log n)$, and thus our desired bound.
\end{proof-sketch}

Note that the polynomials in \Cref{eg:quasi-poly-blowup} and \Cref{eg:exp-blowup} with dense factors were \emph{symmetric}. Bisht and Saxena~\cite{BS2025} answered \Cref{conj:sparse-fac} in the affirmative in this case.

\begin{theorem}[\protect{\cite[Lemma 4.12]{BS2025}}]
    Let $f=gh$ be a factorisation with $g$ being a \emph{symmetric} polynomial in $n$ variables of \emph{individual degree} $d$ and sparsity $\norm{f}=s$. Then, $\norm{g} \leq s^{O(d^2\log d)}$.
\end{theorem}

Furthermore, they used the above result to give a deterministic polynomial-time factoring algorithm~\cite[Theorem 1.2]{BS2025} in this case. For general sparse polynomials, the aforementioned result of Bhattacharjee, Kumar, Ramanathan, Saptharishi, and Saraf~\cite{BKR+2025} gives a subexponential deterministic algorithm as a special case when the depth is two.

\section*{Acknowledgements}
The authors thank the conducive atmosphere of the {\em Workshop on Algebraic Complexity Theory 2023} in University of Warwick to initiate new ideas; and the {\em Workshop on Recent Trends in Computer Algebra 2023} in Institut Henri Poincar\'e, Paris for giving N.S.~the opportunity to chalk out the details of \cite{Sax2023}. C.S.B.~and N.S.~thank the organisers and the participants in Ruhr University Bochum for the interesting discussions in the {\em 8th Workshop on Algebraic Complexity Theory} (\href{https://qi.rub.de/events/wact25/}{WACT 2025}). We also thank Amit Sinhababu and Varun Ramanathan for insightful discussions on Newton iteration and for pointing us to some relevant references.

C.S.B.~thanks the European Union (ERC, CountHom, 101077083) for funding part of this work.

N.S.~thanks the DST-SERB agencies for funding support through the Core Research Grant (CRG/2020/000045) and the J.C.~Bose National Fellowship (JCB/2022/57), as well as the N.~Rama~Rao Chair (2019--) of the Department of CSE, IIT Kanpur.

P.D.~thanks the \emph{Independent Research Fund Denmark} for funding support (FLows 10.46540/3103-00116B), and also acknowledges the support of Basic Algorithms Research Copenhagen (BARC) through the Villum Investigator Grant 54451.

We thank the anonymous reviewers for providing useful feedback that greatly helped in improving the presentation and correctness of the survey.

\printbibliography
\end{document}